\title{\huge \textbf{Confidence intervals for two-stage adaptive designs with subpopulation selection}}

\author{Enyu Li$^{1\ast}$, Nigel Stallard$^{1}$, Ekkehard Glimm$^{2}$, Dominic Magirr$^{2}$, Peter K. Kimani $^{1}$ \\[8pt]
\textit{$^1$Clinical Trials Unit, University of Warwick, Coventry, U.K.} \\
\textit{$^2$Advanced Methodology and Data Science, Novartis Pharma AG, Basel, Switzerland}
\\[8pt]
$^\ast$Email: \href{mailto:Enyu.Li@warwick.ac.uk}{Enyu.Li@warwick.ac.uk}}

\date{ }

\documentclass[hidelinks, 11pt]{article}

\usepackage[margin=0.78in]{geometry}
\usepackage{hyperref}
\usepackage{amssymb, amsmath}
\usepackage{graphicx}
\usepackage{booktabs}
\usepackage{multirow}
\usepackage{makecell}
\usepackage{bm}
\usepackage{caption}
\usepackage{subcaption}
\usepackage{natbib}
\usepackage{pdfpages}
\usepackage{comment}
\usepackage{romannum} 
\usepackage{amsthm}
\usepackage{enumitem}
\usepackage{tikz}
\usepackage{float} 
\usepackage{array}
\usepackage[nameinlink,capitalise]{cleveref}
\usepackage{rotating}
\usepackage{float}

\newcolumntype{C}[1]{>{\centering\arraybackslash}m{#1}}

\newtheorem{theorem}{Theorem}
\newtheorem{proposition}{Proposition}
\newtheorem{lemma}{Lemma}
\newtheorem{corollary}{Corollary}
\newtheorem{definition}{Definition}

\usetikzlibrary{shapes.geometric, arrows.meta, positioning}
\tikzstyle{block} = [rectangle, draw, fill=white!20, text width=8em, text centered, rounded corners, minimum height=3.5em, line width=0.3mm]
\tikzstyle{decision} = [rectangle, draw, fill=white!20, text width=8em, text centered, rounded corners, node distance=1.5cm and -0.3cm, line width=0.3mm, minimum height=3em]
\tikzstyle{line} = [draw, -Latex]

\begin{document}
\pagenumbering{arabic}
\maketitle

\begin{abstract}
We consider clinical trials in which an experimental treatment is compared with a control in pre-specified patient subpopulations. In such settings, adaptive enrichment designs allow the enrolled population to be modified at an interim analysis, with subpopulations selected according to preplanned rules. Since these interim decisions are data-dependent, valid statistical inference must account for them. We focus on constructing confidence intervals for the treatment effect in the selected population. Confidence interval methods that ignore the possibility of population modification may fail to achieve the desired coverage probability. We propose a new approach that constructs confidence intervals with exact $100(1-\alpha)\%$ coverage conditional on the interim decision. Importantly, our method applies to a broad class of adaptive enrichment designs, rather than a single specific design. Our method involves deriving the distribution of the naive estimator of the treatment effect in the selected population conditional on the interim decision and inverting uniformly most accurate unbiased tests to obtain the confidence interval. We provide an efficient computational procedure and show through extensive simulations that the resulting confidence intervals satisfy the theoretical coverage guarantees.
\end{abstract}

\hspace{-1em} {\em Key words}: Adaptive enrichment design; Interval estimation; Precision medicine; Seamless phase II/III design; Subgroup analysis

\section{Introduction}\label{sec:intro}
Recent advances in biomedical research have revealed a plethora of genetic biomarkers which may be related to the success of medical treatments. For example, metastatic breast cancer patients with overexpression of human epidermal growth factor receptor-2 (HER2) are more likely to get desired benefits from HER2 targeted therapies, such as trastuzumab and pertuzumab, as described by \citet{baselga2001herceptin} and \citet{capelan2013pertuzumab}. Subpopulations may also be associated with other baseline characteristics. For example, a meta-analysis by \citet{kirsch2008initial} suggested that a certain class of antidepressants may only benefit patients with higher baseline severity of depression. 

When it is suspected that a treatment may benefit certain subpopulations only, it may be more efficient and ethically justifiable to incorporate an interim analysis into a trial such that the patient population enrolled is allowed to be modified mid-trial based on accrued data using preplanned rules. For example, if the interim analysis gives evidence that the treatment only benefits a certain subpopulation, patient recruitment can subsequently be restricted to this subpopulation. Such designs are known as adaptive enrichment designs \citep{simon2013adaptive}. \citet{rosenblum2011optimizing} and \citet{rosenblum2015adaptive} showed that, with the same expected sample size, adaptive enrichment designs can identify subpopulations that benefit from the experimental treatment and evaluate subpopulation-specific treatment effects more effectively than standard fixed designs.

In an adaptive enrichment design, multiple treatment effects in different subpopulations are considered, and the change of the population enrollment depends on data observed at interim analysis. This introduces statistical challenges on the inference about the treatment effect in the selected subpopulation. Standard hypothesis testing procedures which ignore the adaptive nature of the design may not guarantee the type I error rate. Similarly, in adaptive enrichment designs, naive point estimators may be biased, and naive confidence intervals may fail to achieve nominal coverage probabilities (e.g. \citet{kimani2015estimation} and \citet{magnusson2013group}). Methods have been proposed to address hypothesis testing and point estimation for adaptive enrichment designs, including \citet{brannath2009confirmatory, jenkins2011adaptive, rosenblum2011optimizing, carreras2013shrinkage, kimani2015estimation, robertson2016accounting, kimani2018point}, and \citet{kimani2020point}. However, confidence intervals for adaptive enrichment designs have drawn less attention.

We focus on constructing confidence intervals for the treatment effect in the subpopulations selected at the interim analysis. \citet{european2007reflection} remarked that using an adaptive design implies that the statistical methods control the pre-specified type I error, and that correct estimates and confidence intervals for the treatment effect are available. Hence, construction of confidence intervals, like hypothesis testing and point estimation, is an integral aspect of inference. 

We consider two-stage adaptive enrichment designs. At the interim analysis, a decision is made based on stage 1 data to determine the enrollment criteria for stage 2, with the options being to continue with the full population, any combination of subpopulations, or stop for futility. We develop the method in a broad class of designs involving multiple subpopulations. The designs we consider only allow changes to population enrollment, rather than total sample size, number of treatments, or randomization probabilities.

Our contribution is a new methodology for constructing confidence intervals that achieve exact $100(1-\alpha)\%$ conditional coverage given the interim decision, i.e., the selected population, as formalized in \cref{sec:problem} and \cref{sec:construction}. We first derive the conditional distribution of the naive treatment effect estimator in the selected population, and show that it is a sufficient statistic under the conditional model. Building on the conditional inference framework of \citet{lehmann1986testing}, we then obtain the associated conditional uniformly most accurate unbiased confidence interval. In adaptive enrichment designs, however, the conditional distributions are analytically complex, and the resulting confidence intervals do not have closed-form expressions, and their calculation is complex. To address this, we develop an efficient numerical procedure that only requires evaluation of a definite integral of a well-defined probability density and root-finding for a monotone continuous function, both of which are readily implemented in standard statistical software. To the best of our knowledge, no existing confidence interval procedure for adaptive enrichment designs provides exact conditional coverage, and our approach is the first to derive and compute the uniformly most accurate unbiased confidence interval in this adaptive design setting.

\section{Related Work} \label{sec:relatedwork}
A number of methods have been proposed to construct confidence intervals for adaptive designs involving treatment selection, subpopulation selection, and early stopping \citep{posch2005testing, sampson2005drop, brannath2006estimation, bowden2008unbiased, koopmeiners2012conditional, magnusson2013group, rosenblum2013confidence, robertson2016accounting, kimani2020point}. In general, these approaches control the overall coverage probability at the nominal level $100(1-\alpha)\%$. However, for the class of adaptive enrichment designs considered in \cref{sec:problem}, these existing procedures are either not applicable or do not provide accurate coverage. We give more details below.

\citet{rosenblum2013confidence} constructed confidence intervals for adaptive enrichment designs by inflating the width of the naive confidence intervals that ignores adaptivity. Their method guarantees at least $100(1-\alpha)\%$ overall coverage, defined as a weighted average of the conditional coverage probabilities across all possible interim decisions. In contrast, the method proposed in this paper provides the stronger guarantee of exact conditional coverage given the interim decision. Furthermore, our framework accommodates more general enrichment designs, including those with futility stopping and multiple subpopulations, in which the direct application of the method of \citet{rosenblum2013confidence} is not straightforward.

\citet{magnusson2013group} considered group sequential designs with subpopulation selection and proposed double bootstrap confidence intervals for the treatment effect in the selected population. Following the bootstrap bias-correction strategy of \citet{davison1997bootstrap}, they iteratively estimated the selection bias of the naive estimator to obtain a bias-corrected confidence interval. Although their aim was to achieve conditional coverage at level $100(1-\alpha)\%$, their simulation results indicated substantial under-coverage in many settings.

\citet{magirr2013simultaneous} presented a general procedure for constructing one-sided simultaneous confidence intervals, which guarantees at least $100(1-\alpha)\%$ overall simultaneous coverage, for a broad class of adaptive designs. \citet{kimani2020point} focused on an adaptive enrichment design and extended this approach to obtain two-sided confidence intervals, referred to as duality confidence intervals. However, duality confidence intervals can be highly conservative and, in certain cases, may be non-informative (e.g. covering the entire parameter space).

\citet{sampson2005drop} studied confidence intervals for drop-the-loser designs, where multiple treatments are compared and the best-performing treatment in stage 1 is selected for continuation in stage 2. They derived the conditional distribution of the naive treatment effect estimator given that the treatment is selected, and constructed confidence intervals by inverting two one-sided level $\alpha/2$ uniformly most powerful tests \citep{lehmann1986testing}. Our work extends this conditional inference perspective by constructing the uniformly most accurate unbiased confidence interval which is obtained by inverting the two-sided level $\alpha$ uniformly most powerful unbiased test. Moreover, in contrast to the design in \citet{sampson2005drop}, our procedure applies to adaptive enrichment designs with multiple subpopulations and stopping for futility.

\section{Problem Definition and Notation} \label{sec:problem}

\subsection{Data description, notation, and assumptions} \label{sec:notation}
Suppose the full patient population can be partitioned into $k$ disjoint subpopulations. Denote the subpopulations 1 to $k$ and the full population by $\mathcal{S}_1$, $\dots$, $\mathcal{S}_k$, and $\mathcal{F}$, respectively. We write $\mathcal{K} = \{1,\dots, k\}$ as the index set for the $k$ subpopulations. We have $\bigcup_{m \in \mathcal{K}} \mathcal{S}_m = \mathcal{F}$ and $ \mathcal{S}_m \bigcap \mathcal{S}_{m'} = \emptyset$ for any $m, m' \in \mathcal{K}, \ m \neq m'$. For each $m \in \mathcal{K}$, let $p_m$ denote the proportion of patients from subpopulation $\mathcal{S}_m$ in the full population where $\sum_{m \in \mathcal{K}} p_m =1$. We focus on adaptive enrichment designs with two stages. Let $n$ denote the total sample size, and let $n_j = w_j n$ be the sample size in stage $j \in \{1,2\}$, where $w_1 + w_2 = 1$. In the first stage where patients are enrolled from the full population, for each subpopulation $\mathcal{S}_m$, we assume that the proportion of patients from $\mathcal{S}_m$ is the same as the subpopulation proportion $p_m$. After the interim analysis, stage 2 enrollment may be restricted to any combination of subpopulations that is likely to benefit, which is called enrichment. Let $\mathcal{K}^{(2)}$ denote the index set for the selected subpopulations, where $\mathcal{K}^{(2)} \subseteq \mathcal{K}$. Define the selected subpopulations by $\mathcal{S}_{\mathcal{K}^{(2)}} = \bigcup_{m \in \mathcal{K}^{(2)}} \mathcal{S}_m$. In the second stage, we assume that patients are enrolled from subpopulation $m \in \mathcal{K}^{(2)}$ in proportion $p_m/p_{\mathcal{K}^{(2)}}$, where $p_{\mathcal{K}^{(2)}} = \sum_{m \in \mathcal{K}^{(2)}} p_m$. Within each enrolled subpopulation during a given stage, we assume that half of the patients are randomized to the treatment arm and half to the control arm.

For each patient $i \in  \{1,\dots,n\}$, we collect the following data, $(S_i, J_i, A_i, Y_i)$, where $S_i \in \mathcal{K}$ denotes the subpopulation index, $J_i \in \{1,2\}$ is the stage when the patient is enrolled, $A_i \in \{0,1\}$ is the treatment arm indicator ($0$ for control, $1$ for treatment), and $Y_i \in \mathbb{R}$ is the outcome. We assume that outcomes are normally distributed, with mean $\mu_{ma}$ in subpopulation $\mathcal{S}_m$ and study arm $a \in \{0,1\}$, and with common variance $\sigma^2$. Throughout the paper, we assume larger outcome values correspond to greater treatment effect. For each subpopulation $\mathcal{S}_m$, define the true treatment effect in $\mathcal{S}_m$ by $\Delta_m$, where $\Delta_m = \mu_{m1} - \mu_{m0}$. For any combination of subpopulations $\mathcal{S}_{\mathcal{K}^{(2)}}$, we define the true treatment effect in $\mathcal{S}_{\mathcal{K}^{(2)}}$ by $\Delta_{{\mathcal{K}^{(2)}}}$, where $\Delta_{\mathcal{K}^{(2)}} = ({\sum_{m \in \mathcal{K}^{(2)}} p_m\Delta_m})/{p_{\mathcal{K}^{(2)}}}$.Accordingly, the true treatment effect in the full population $\mathcal{F}$ is given by $\Delta_{\mathcal{K}} = \sum_{m \in \mathcal{K}} p_m \Delta_m$.

Throughout the paper, we assume $n_1$ and $n_2$ are fixed at the beginning of the trial. In addition, we suppose $p_m$ for each $m \in \mathcal{K}$ and $\sigma$ are known.

\subsection{Definition of test statistics and naive confidence interval}

For any combination of subpopulations $\mathcal{S}_{\mathcal{K}^{(2)}}$ and each stage $j \in \{1,2\}$ in which there is an enrollment in $\mathcal{S}_m$ for any $m \in \mathcal{K}^{(2)}$, define the sample mean difference between treatment versus control by
\begin{equation*}
	\hat{\Delta}_{\mathcal{K}^{(2)}}^{(j)} = \frac{\sum_{\{i:S_i \in \mathcal{K}^{(2)}, J_i = j, A_i = 1\}}Y_i}{\vert \{i:S_i \in \mathcal{K}^{(2)}, J_i = j, A_i = 1\} \vert} - \frac{\sum_{\{i:S_i \in \mathcal{K}^{(2)}, J_i = j, A_i = 0\}}Y_i}{\vert \{i:S_i \in \mathcal{K}^{(2)}, J_i = j, A_i = 0\} \vert} .
\end{equation*}

In the following, we define the naive point estimator for the true treatment effect to be the estimator obtained by pooling data of all patients enrolled in both stage 1 and 2. If the trial continues with the combination of subpopulations $\mathcal{S}_{\mathcal{K}^{(2)}}$, we define the naive point estimator for the true treatment effect in $\mathcal{S}_{\mathcal{K}^{(2)}}$ by
\begin{equation*}
	\hat{\Delta}_{\mathcal{K}^{(2)}} = \frac{\sum_{\{i:S_i \in \mathcal{K}^{(2)}, A_i = 1\}}Y_i}{\vert \{i:S_i \in \mathcal{K}^{(2)}, A_i = 1\} \vert} - \frac{\sum_{\{i:S_i \in \mathcal{K}^{(2)}, A_i = 0\}}Y_i}{\vert \{i:S_i \in \mathcal{K}^{(2)}, A_i = 0\} \vert}.
\end{equation*}

We define the naive $100(1-\alpha)\%$ confidence interval that ignores the adaptive nature of the design as
\begin{equation*}
	\left[\,\hat{\Delta}_{ \mathcal{K}^{(2)} } + \Phi^{-1}({\alpha}/{2}) \frac{2\sigma}{\sqrt{p_{\mathcal{K}^{(2)}}n_1 + n_2}},\; \hat{\Delta}_{ \mathcal{K}^{(2)} } + \Phi^{-1}(1 - {\alpha}/{2}) \frac{2\sigma}{\sqrt{p_{\mathcal{K}^{(2)}}n_1 + n_2}}\right]
\end{equation*}

For notational simplicity, for single subpopulation $\mathcal{S}_m$, we write $\hat{\Delta}_{m}^{(j)}$ and $\hat{\Delta}_{m}$ instead of $\hat{\Delta}_{\{m\}}^{(j)}$ and $\hat{\Delta}_{\{m\}}$, respectively, throughout the paper.

\subsection{Decision rule} \label{sec:decision}

Denote the stage 1 data by $\mathcal{X}^{(1)} = \{(S_i, J_i, A_i, Y_i): i \in \{1,\dots,n\},\, J_i = 1\}$. Suppose that $\mathcal{X}^{(1)}$ is available for the interim analysis. Let $D: \mathcal{X}^{(1)} \mapsto \mathcal{K}^{(2)}$ denote the decision rule, which maps the stage 1 data to a unique index set $\mathcal{K}^{(2)} \subseteq \mathcal{K}$ representing the subpopulations to be enrolled in stage 2, with $D(\mathcal{X}^{(1)}) = \emptyset$ indicating that the trial stops for futility after stage 1.

We introduce a general class of interim decision rules, denoted by $\mathcal{D}$, for two-stage adaptive enrichment designs, to which the theoretical results in \cref{sec:construction} apply for constructing the proposed confidence intervals. Decision rules in $\mathcal{D}$ satisfy the following condition: for each non-empty realization of $\mathcal{K}^{(2)}$, there exist random thresholds $L_{\mathcal{K}^{(2)}}$ and $U_{\mathcal{K}^{(2)}}$ such that
\begin{align*}
	\{\mathcal{X}^{(1)}: D(\mathcal{X}^{(1)}) = \mathcal{K}^{(2)} \} = \{\mathcal{X}^{(1)}: L_{\mathcal{K}^{(2)}} < \hat{\Delta}^{(1)}_{\mathcal{K}^{(2)}} < U_{\mathcal{K}^{(2)}}\}.
\end{align*}
Here, $L_{\mathcal{K}^{(2)}}$ and $U_{\mathcal{K}^{(2)}}$ may also take values in $\overline{\mathbb{R}} = \mathbb{R} \cup \{-\infty,+\infty\}$. In addition, $L_{\mathcal{K}^{(2)}}$ and $U_{\mathcal{K}^{(2)}}$ are assumed to be independent of $\hat{\Delta}^{(1)}_{\mathcal{K}^{(2)}}$. The inequalities may be taken as either strict or non-strict, as the distinction does not affect the validity of the proposed method. Intuitively, a subpopulation is selected if its stage~1 sample mean $\hat{\Delta}^{(1)}_{\mathcal{K}^{(2)}}$ lies within a range whose limits are determined by random variables independent of $\hat{\Delta}^{(1)}_{\mathcal{K}^{(2)}}$.

If the trial continues with the subpopulations $\mathcal{S}_{\mathcal{K}^{(2)}}$, we aim to construct a $100(1-\alpha)\%$ confidence interval for $\Delta_{\mathcal{K}^{(2)}}$, whose conditional coverage probability satisfies
\begin{equation*}
	\Pr\!\Big( \Delta_{\mathcal{K}^{(2)}} \in \text{CI}_{1-\alpha}(\Delta_{\mathcal{K}^{(2)}})
	\,\big|\, D(\mathcal{X}^{(1)}) = \mathcal{K}^{(2)} \Big) = 1-\alpha.
\end{equation*}

Although the class $\mathcal{D}$ is defined above in terms of constraints on the stage 1 estimator, our framework also applies to certain interim decision rules that additionally depend on auxiliary statistics, provided that the resulting selection event can be characterized through joint constraints on $\hat{\Delta}^{(1)}_{\mathcal{K}^{(2)}}$ and the auxiliary statistic. Under this condition, the proposed confidence interval in \cref{sec:construction} remain valid. For completeness, an example involving such an auxiliary statistic as in \citet{kimani2015estimation} is provided in Appendix C, where we illustrate this extension.

\subsection{Example designs with decision rule in class $\mathcal{D}$}\label{sec:designs}

The class $\mathcal{D}$ encompasses a broad range of interim decision rules arising in adaptive enrichment designs, including those proposed in the literature \citep{rosenblum2013confidence, kimani2015estimation, kimani2018point}. We first describe the decision rule of \citet{rosenblum2013confidence} and the rule proposed in this paper, which are used in the worked example (\cref{sec:example}) and simulation study (\cref{sec:sim}). Additional illustrations corresponding to the decision rules of \citet{kimani2015estimation, kimani2018point} are provided in Appendix C.

\subsubsection{The interim decision rule proposed by \citet{rosenblum2013confidence}} \label{sec:design1}
\citet{rosenblum2013confidence} considered the two subpopulations case, i.e., $\mathcal{K} = \{1,2\}$ and proposed a decision rule based on $Z$-statistics. Define the stage 1 $Z$-statistics in the full population $\mathcal{F}$, $\mathcal{S}_1$, and $\mathcal{S}_2$, respectively, as
\begin{equation*}
	Z_{\mathcal{K}}^{(1)}= \frac{\hat{\Delta}_{\mathcal{K}}^{(1)}}{2\sigma / \sqrt{n_1}}, \qquad Z_{1}^{(1)} = \frac{\hat{\Delta}_{1}^{(1)}}{2\sigma / \sqrt{p_1 n_1}}, \qquad Z_{2}^{(1)}= \frac{\hat{\Delta}_{2}^{(1)}}{2\sigma / \sqrt{p_2 n_1}}.
\end{equation*}
Let $D_1$ denote the decision rule shown in \cref{fig:decision1}. For a fixed $Z_*$, if $Z_{\mathcal{K}}^{(1)}> Z_*$, the trial continues with the full population $\mathcal{F}$, i.e., $D_1(\mathcal{X}^{(1)}) = \{1,2\}$. Otherwise, the trial enriches to the subpopulation with the greater standardized stage 1 sample mean difference such that $D_1(\mathcal{X}^{(1)}) = \{ \arg\max_{m \in \{1,2\}} Z_m^{(1)}\}$ with ties broken arbitrarily. Here, we assume the trial enriches to $\mathcal{S}_1$ when $Z_{1}^{(1)} = Z_{2}^{(1)}$. In this design, the possible values of $D_1(\mathcal{X}^{(1)})$ are $\{1,2\}$, $\{1\},$and $ \{2\}$ corresponding respectively to continuation in the full population $\mathcal{F}$, enrichment in $\mathcal{S}_1$, and enrichment in $\mathcal{S}_2$. 

\begin{figure}[!htbp]
	\begin{center}
		\begin{tikzpicture}[node distance=1.5cm and 0cm, auto]
			\node (start) [decision] {$Z_{\mathcal{K}}^{(1)}> Z_*$};
			\node (outcome1) [block, below left=2cm and 0cm of start] {Enrol from $\mathcal{F}$};
			\node (outcome2) [block, below right=2cm and 0cm of start] {Enrol from $\mathcal{S}_{\arg\max_{m \in \{1,2\}} Z_m^{(1)}}$};
			\path [line] (start) -- node [midway, left] {Yes \hspace{0.5em}} (outcome1);
			\path [line] (start) -- node [midway, right] {\hspace{0.5em}  No} (outcome2);
		\end{tikzpicture}
	\end{center}
	\caption{Schematic diagram of the interim decision rule $D_1$ (proposed by \citet{rosenblum2013confidence})}
	\label{fig:decision1}
\end{figure}
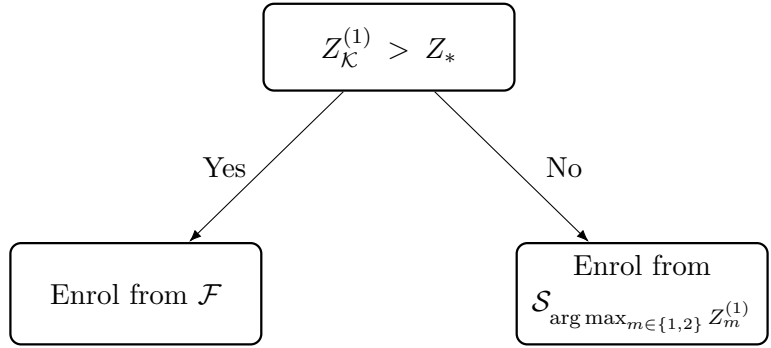

In Appendix B.1, we show that $D_1$ induces the following partition of the stage 1 sample space and hence lies in class $\mathcal{D}$.
\begin{align*}
	&\{\mathcal{X}^{(1)}: D_1(\mathcal{X}^{(1)}) = \mathcal{K}\} = \{\mathcal{X}^{(1)}: \frac{2\sigma}{\sqrt{n_1}} Z_* < \hat{\Delta}_{\mathcal{K}}^{(1)} < +\infty \} \\
	&\{\mathcal{X}^{(1)}: D_1(\mathcal{X}^{(1)}) = \{1\}\} = \{\mathcal{X}^{(1)}: \sqrt{\frac{p_2}{p_1}} \hat{\Delta}_{2}^{(1)} \leq \hat{\Delta}_{1}^{(1)} \leq \frac{2\sigma}{p_1 \sqrt{n_1}}Z_* - \frac{p_2}{p_1}\hat{\Delta}_{2}^{(1)} \} \\
	&\{\mathcal{X}^{(1)}: D_1(\mathcal{X}^{(1)}) = \{2\}\} = \{\mathcal{X}^{(1)}: \sqrt{\frac{p_1}{p_2}} \hat{\Delta}_{1}^{(1)} < \hat{\Delta}_{2}^{(1)} \leq \frac{2\sigma}{p_2 \sqrt{n_1}}Z_* - \frac{p_1}{p_2}\hat{\Delta}_{1}^{(1)} \}
\end{align*}

\subsubsection{An interim decision rule with stopping for futility} \label{sec:design2}
We also consider the two subpopulations case ($\mathcal{K} = \{1,2\}$) and propose an interim decision rule with stopping for futility, which may terminate the trial when the stage 1 observations are not promising. The decision rule, $D_2$, is illustrated in \cref{fig:decision2}. If the stage 1 sample mean difference in the full population $\mathcal{F}$ exceeds the pre-specified threshold $\Delta_*$, the trial continues with $\mathcal{F}$ in stage 2. Otherwise, if the larger of the stage 1 sample mean differences in $\mathcal{S}_1$ and $\mathcal{S}_2$ exceeds $\Delta_*$, the trial enriches to the corresponding subpopulation. If none of these stage 1 sample mean differences exceeds $\Delta_*$, the trial stops for futility. The possible values of $D_2(\mathcal{X}^{(1)})$ are $\{1,2\}$, $\{1\}, \{2\}$, and $\emptyset$, corresponding respectively to continuation in the full population $\mathcal{F}$, enrichment in $\mathcal{S}_1$, enrichment in $\mathcal{S}_2$, and stopping for futility in stage 2. 

\begin{figure}[!htbp]
	\begin{center}
		\begin{tikzpicture}[node distance=1.5cm and 0cm, auto]
			\node (start) [decision] {$\hat{\Delta}_{\mathcal{K}}^{(1)}> \Delta_*$};
			\node (outcome1) [block, below left=3.55cm and 0cm of start] {Enrol from $\mathcal{F}$};
			\node (dec2) [decision, below right=1cm and 0cm of start, text width=15em] {$\max(\hat{\Delta}_{1}^{(1)}, \hat{\Delta}_{2}^{(1)}) > \Delta_*$};
			\node (outcome2) [block, below left=1.3cm and -3.5cm of dec2, text width=12em] {Enrol from $\mathcal{S}_{\arg\max_{m \in \{1,2\}}\hat{\Delta}_{m}^{(1)}}$};
			\node (outcome3) [block, below right=1.3cm and -1cm of dec2] {Stop} ;
			
			\path [line] (start) -- node [midway, left] {Yes \hspace{0.5em}} (outcome1);
			\path [line] (start) -- node [midway, right] {\hspace{0.5em}  No} (dec2);
			\path [line] (dec2) -- node [midway, left] {Yes \hspace{0.5em}} (outcome2);
			\path [line] (dec2) -- node [midway, right] {\hspace{0.5em} No} (outcome3);
		\end{tikzpicture}
	\end{center}
	\caption{Schematic diagram of the interim decision rule $D_2$}
	\label{fig:decision2}
\end{figure}
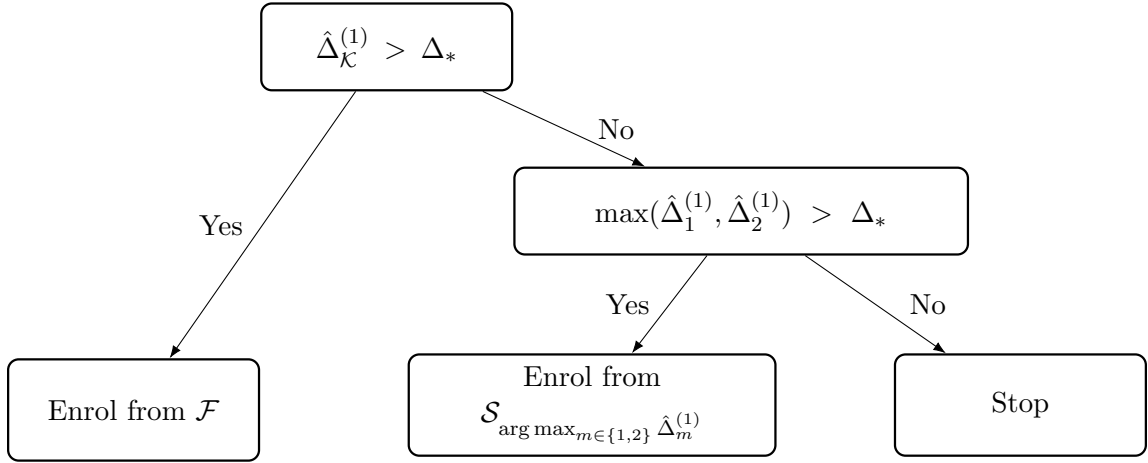

In Appendix B.2, we demonstrate that $D_2$ yields the following partition of the stage 1 sample space and also lies in class $\mathcal{D}$.
\begin{align*}
	&\{\mathcal{X}^{(1)}: D_2(\mathcal{X}^{(1)}) = \mathcal{K}\} = \{\mathcal{X}^{(1)}: \Delta_* < \hat{\Delta}^{(1)}_{\mathcal{K}} < +\infty\} \\
	&\{\mathcal{X}^{(1)}: D_2(\mathcal{X}^{(1)}) = \{1\}\} =  \{\mathcal{X}^{(1)}: \Delta_*<\hat{\Delta}_{1}^{(1)} \leq (\Delta_* - p_2 \hat{\Delta}_{2}^{(1)})/p_1\} \\
	&\{\mathcal{X}^{(1)}: D_2(\mathcal{X}^{(1)}) = \{2\}\} =  \{\mathcal{X}^{(1)}: \Delta_*<\hat{\Delta}_{2}^{(1)} \leq (\Delta_* - p_1 \hat{\Delta}_{1}^{(1)})/p_2\}
\end{align*}

\subsection{Co-primary analysis}\label{sec:co-primary}
Following \citet{jenkins2011adaptive}, when the full population continues to stage 2, investigators may still be interested in the treatment effect in a particular subpopulation and in how this effect compares with that in the full population, either as an enhancement or an attenuation. Therefore, a co-primary analysis can be included. For example, in the case of two subpopulations, when the full population $\mathcal{F}$ continues to stage 2, $100(1-\alpha)\%$ confidence intervals for $\Delta_1$ and $\Delta_2$ may also be constructed, with conditional coverage

\begin{align*}
	\Pr\!\Big( \Delta_{1} \in \text{CI}_{1-\alpha}(\Delta_{1})
	\,\big|\, D(\mathcal{X}^{(1)}) = \{1,2\} \Big) = 1-\alpha \ \intertext{and} \  \Pr\!\Big( \Delta_{2} \in \text{CI}_{1-\alpha}(\Delta_{2})
	\,\big|\, D(\mathcal{X}^{(1)}) = \{1,2\} \Big) = 1-\alpha.
\end{align*}

\section{Construction of Conditional Uniformly Most Accurate Unbiased Confidence Intervals}\label{sec:construction}

\subsection{General method and theorem}\label{sec:method}

\begin{theorem} \label{theorem}
	Suppose $\hat{\Delta}^{(1)} \sim \mathcal{N}(\Delta, \sigma_{(1)}^2)$ and $\hat{\Delta}^{(2)} \sim \mathcal{N}(\Delta, \sigma_{(2)}^2)$, where $\sigma_{(1)}^2$ and $\sigma_{(2)}^2$ are known. Let $\tau_{(1)} = 1/{\sigma_{(1)}^2},\ \tau_{(2)} = 1/{\sigma_{(2)}^2},\ \hat{\Delta}=\tau_{(1)}/(\tau_{(1)} + \tau_{(2)})\hat{\Delta}^{(1)} + \tau_{(2)}/(\tau_{(1)} + \tau_{(2)})\hat{\Delta}^{(2)}$ and $\sigma_{(12)}^2 = \sigma_{(1)}^2 \sigma_{(2)}^2/ (\sigma_{(1)}^2 + \sigma_{(2)}^2)$.
	\begin{enumerate}
		\item[(i)] The conditional distribution of $\hat{\Delta}$ given $l < \hat{\Delta}^{(1)}  < u$, where $l$ and $u$ are constants in $\overline{\mathbb{R}} = \mathbb{R} \cup \{-\infty, +\infty\}$, is
		\begin{equation*}
			f_{\Delta}(\hat{\Delta} \mid l < \hat{\Delta}^{(1)} < u) = \frac{1}{\sigma_{(12)}} \phi(\frac{\hat{\Delta} - \Delta}{\sigma_{(12)}}) \frac{\Phi(\frac{u - \hat{\Delta}}{(\sigma_{(1)}/\sigma_{(2)}) \sigma_{(12)}}) - \Phi(\frac{l - \hat{\Delta}}{(\sigma_{(1)}/\sigma_{(2)}) \sigma_{(12)}})}{\Phi(\frac{u - \Delta}{\sigma_{(1)}}) - \Phi(\frac{l - \Delta}{\sigma_{(1)}})},
		\end{equation*}
		where $\phi$ and $\Phi$ are probability density function and cumulative distribution function of standard normal distribution. Furthermore, $\hat{\Delta}$ is a sufficient statistic for $\Delta$.
		\item[(ii)]  The conditional level $\alpha$ uniformly most powerful unbiased test for $H: \Delta = \Delta_0$, given $l < \hat{\Delta}^{(1)}  < u$, exists. The critical function $\psi_{\Delta_0}(\hat{\Delta})$ for testing $H$ is
		\begin{equation*}
			\psi_{\Delta_0}(\hat{\Delta}) = 
			\begin{cases}
				1 & \text{when } \hat{\Delta} < C_1(\Delta_0) \text{ or } \hat{\Delta} > C_2(\Delta_0), \\
				0 & \text{when } C_1(\Delta_0) \leq \hat{\Delta} \leq C_2(\Delta_0).
			\end{cases}
		\end{equation*}
		where $C_1(\Delta_0)$ and $C_2(\Delta_0)$ are determined by
		\begin{equation*}
			\begin{cases}
				\mathbb{E}_{\Delta_0}(\psi_{\Delta_0}( \hat{\Delta})) = \alpha, \\
				\mathbb{E}_{\Delta_0}( \hat{\Delta}\psi_{\Delta_0}( \hat{\Delta})) = \alpha \mathbb{E}_{\Delta_0}( \hat{\Delta}).
			\end{cases}
		\end{equation*}
		\item[(iii)]  The inverse functions, $C_1^{-1}$ and $C_2^{-1}$, exist. By inverting the uniformly most powerful unbiased tests, the conditional $100(1 - \alpha)\%$ uniformly most accurate unbiased confidence interval for $\Delta$, $(\underline{\Delta},\ \overline{\Delta})$, is given by
		\begin{equation*}
			\underline{\Delta} = C_2^{-1}(\hat{\Delta}),\;\;
			\overline{\Delta} = C_1^{-1}(\hat{\Delta}).
		\end{equation*}
	\end{enumerate}
\end{theorem}

\proof
See Appendix A.1.
\endproof

We denote our proposed conditional uniformly most accurate unbiased confidence interval by C-UMAU confidence interval.

\begin{proposition} \label{proposition:2}
	The C-UMAU confidence interval for $\Delta$ has the following properties:
	\begin{enumerate}
		\item[(i)]  Exact conditional coverage: $\Pr \big(\Delta \in (\underline{\Delta},\ \overline{\Delta}) \mid l < \hat{\Delta}^{(1)} < u \big) = 1 - \alpha$.
		\item[(ii)]  Conditional unbiased: $\Pr \big(\Delta' \in (\underline{\Delta},\ \overline{\Delta}) \mid l < \hat{\Delta}^{(1)} < u \big) \leq 1 - \alpha$, for any $\Delta' \neq \Delta$.
		\item[(iii)]  Uniformly most accurate among unbiased intervals: among all confidence intervals with conditional coverage at least $1 - \alpha$, the C-UMAU interval minimizes $\Pr \big(\Delta' \in (\underline{\Delta},\ \overline{\Delta}) \mid l < \hat{\Delta}^{(1)} < u \big)$ for all $\Delta' \neq \Delta$.
	\end{enumerate}
\end{proposition}

\proof
See Appendix A.2.
\endproof

\subsection{Numerical method for calculating the C-UMAU confidence interval}\label{sec:numerical}

Suppose $\hat{\Delta}$ is distributed as in \cref{theorem}(i). We next provide a detailed procedure for constructing the $100(1-\alpha)\%$ C-UMAU confidence interval by inverting the conditional uniformly most powerful tests.

Step 1: Construct the conditional uniformly most powerful unbiased test for $H:\Delta = \Delta_0$.
\begin{enumerate}
	\item For any $\Delta_0 \in \mathbb{R}$, compute $\mathbb{E}_{\Delta_0}(\hat{\Delta})$ using the result from Appendix A.3.
	\item Solve the equation
	\begin{equation*}
		\int_{c_1}^{c_2(c_1)} t \, f_{\Delta_0}(t) \, dt = (1-\alpha)\, \mathbb{E}_{\Delta_0}(\hat{\Delta})
	\end{equation*}
	for $c_1$, where $c_2(c_1) = F^{-1}(F(c_1) + 1 - \alpha)$. We prove $\int_{c_1}^{c_2(c_1)} t \, f_{\Delta_0}(t) \, dt$ is continuous and strictly increasing on $c_1$ in Appendix A.4. Accordingly, we can use any numerical root-finding method such as \texttt{uniroot} in \textsf{R}.
	\item Let $c_1^*$ denote the solution, and define
	\[
	C_1(\Delta_0) = c_1^*, \qquad 
	C_2(\Delta_0) = c_2(c_1^*).
	\]
	\item Then $[C_1(\Delta_0),\, C_2(\Delta_0)]$ is the acceptance region for the conditional uniformly most powerful test in \cref{theorem}(ii).
\end{enumerate}

Step 2: Invert the test to obtain the confidence interval.
\begin{enumerate}
	\item Treat $C_1(\Delta_0)$ and $C_2(\Delta_0)$ as functions of $\Delta_0$. Given the observed $\hat{\Delta}$, solve
	\begin{equation*}
		C_2(\underline{\Delta}) = \hat{\Delta} \quad  \mathrm{and} \quad C_1(\overline{\Delta}) = \hat{\Delta}
	\end{equation*} 
	for the lower limit $\underline{\Delta}(\hat{\Delta})$ and the upper limit $\overline{\Delta}(\hat{\Delta})$, respectively. We prove that $C_1(\Delta_0)$ and $C_2(\Delta_0)$ are continuous and strictly increasing in $\Delta_0$ in Appendix A.5. Then, both roots can be found using \texttt{uniroot} in \textsf{R}.
	\item $\left[\,\underline{\Delta}(\hat{\Delta}),\; \overline{\Delta}(\hat{\Delta})\,\right] $ is the C-UMAU confidence interval.
\end{enumerate}

\subsection{Application on two-stage enrichment designs with interim decision rules in $\mathcal{D}$} \label{sec:application1}
First, we consider the conditional probability density of $\hat{\Delta}_{ \mathcal{K}^{(2)} }$ given the interim decision $D(\mathcal{X}^{(1)}) =  \mathcal{K}^{(2)}$ and the realization of $L_{\mathcal{K}^{(2)}}$ and $U_{\mathcal{K}^{(2)}}$.
\begin{align*}
	& f(\hat{\Delta}_{ \mathcal{K}^{(2)} } \mid D(\mathcal{X}^{(1)}) =  \mathcal{K}^{(2)}, L_{\mathcal{K}^{(2)}} = l_{\mathcal{K}^{(2)}}, U_{\mathcal{K}^{(2)}} = u_{\mathcal{K}^{(2)}}) \\
	= & f(\hat{\Delta}_{\mathcal{K}^{(2)}} \mid l_{\mathcal{K}^{(2)}} < \hat{\Delta}^{(1)}_{\mathcal{K}^{(2)}}< u_{\mathcal{K}^{(2)}}) \\
	= & \frac{\phi(\frac{\hat{\Delta}_{\mathcal{K}^{(2)}} - \Delta_{\mathcal{K}^{(2)}}}{2\sigma/\sqrt{p_{\mathcal{K}^{(2)}}n_1 + n_2}})}{2\sigma/\sqrt{p_{\mathcal{K}^{(2)}}n_1 + n_2}} \frac{\Phi(\frac{l_{\mathcal{K}^{(2)}} - \hat{\Delta}_{\mathcal{K}^{(2)}} }{\sqrt{n_2/(p_{\mathcal{K}^{(2)}}n_1)} \ 2\sigma/\sqrt{n_{\mathcal{K}^{(2)}}}}) - \Phi(\frac{u_{\mathcal{K}^{(2)}} - \hat{\Delta}_{\mathcal{K}^{(2)}} }{\sqrt{n_2/(p_{\mathcal{K}^{(2)}}n_1)} \ 2\sigma/\sqrt{n_{\mathcal{K}^{(2)}}}})}{\Phi(\frac{l_{\mathcal{K}^{(2)}} - \Delta_{\mathcal{K}^{(2)}}}{2\sigma/\sqrt{p_{\mathcal{K}^{(2)}}n_1}}) - \Phi(\frac{u_{\mathcal{K}^{(2)}} - \Delta_{\mathcal{K}^{(2)}}}{2\sigma/\sqrt{p_{\mathcal{K}^{(2)}}n_1}})}.
\end{align*}
The conditional distribution of $\hat{\Delta}_{ \mathcal{K}^{(2)} }$ follows the form in \cref{theorem}(i). Hence, $\hat{\Delta}_{ \mathcal{K}^{(2)} }$ is a sufficient statistic for estimating $\Delta_{ \mathcal{K}^{(2)} }$. Then, we can use the procedure to calculate the C-UMAU confidence interval as described in \cref{sec:numerical}. Since this procedure provides the C-UMAU confidence interval for every realization of $L_{\mathcal{K}^{(2)}}$ and $U_{\mathcal{K}^{(2)}}$, it also constitutes the C-UMAU confidence interval for $\Delta_{ \mathcal{K}^{(2)} }$ given $D(\mathcal{X}^{(1)}) =  \mathcal{K}^{(2)}$ (Theorem 4.4.1 and Lemma 5.5.1 in \citet{lehmann1986testing}). Examples with specific designs will be given in \cref{sec:example} and \cref{sec:sim}.

\section{A computationally simpler alternative based on inversion of two one-sided tests}

The proposed C-UMAU confidence interval is obtained by inverting two-sided conditional uniformly most powerful unbiased tests. In this subsection, we introduce a computationally simpler alternative to the proposed C-UMAU confidence interval, constructed via inversion of two conditional one-sided uniformly most powerful tests. We refer to this construction as the C-TOST (conditional two one-sided tests) confidence interval. By construction, the C-TOST confidence interval guarantees exact conditional coverage, and its lower and upper limits correspond to the conditional uniformly most accurate confidence bounds as defined in \citet{lehmann1986testing}. However, unlike the C-UMAU confidence interval, the resulting interval does not in general satisfy the uniformly most accurate unbiased criterion and is therefore theoretically suboptimal. We include the C-TOST confidence interval as a computationally convenient alternative for comparison.

\begin{definition}
	Let $F_{\Delta}(\hat{\Delta} \mid l < \hat{\Delta}^{(1)} < u)$ denote the conditional cumulative distribution function of $\hat{\Delta}$ given in \cref{theorem}(i). The $100(1-\alpha)\%$ C-TOST confidence interval for $\Delta$, denoted by $(\underline{\Delta}^*,\ \overline{\Delta}^*)$, is defined by
	\begin{align*}
		& F_{\underline{\Delta}^*}(\hat{\Delta} \mid l < \hat{\Delta}^{(1)} < u) = 1 - \frac{\alpha}{2}, \\
		& F_{\overline{\Delta}^*}(\hat{\Delta} \mid l < \hat{\Delta}^{(1)} < u) = \frac{\alpha}{2}.
	\end{align*}
\end{definition}

For two-stage adaptive enrichment designs with interim decision rules in the class $\mathcal{D}$, the conditional probability density function $f(\hat{\Delta}_{ \mathcal{K}^{(2)} } \mid D(\mathcal{X}^{(1)}) =  \mathcal{K}^{(2)}, L_{\mathcal{K}^{(2)}} = l_{\mathcal{K}^{(2)}}, U_{\mathcal{K}^{(2)}} = u_{\mathcal{K}^{(2)}})$ is given in \cref{sec:application1}. For any fixed $\Delta_{\mathcal{K}^{(2)}}$ and observed $\hat{\Delta}_{\mathcal{K}^{(2)}}$, the corresponding cumulative distribution function $F_{\Delta_{\mathcal{K}^{(2)}}}(\hat{\Delta}_{\mathcal{K}^{(2)}} \mid D(\mathcal{X}^{(1)}) = \mathcal{K}^{(2)}, l_{\mathcal{K}^{(2)}}, u_{\mathcal{K}^{(2)}})$ can be computed numerically using \texttt{integrate} in \textsf{R}. Then, the C-TOST confidence limits can be obtained by solving the defining equations using a one-dimensional root-finding algorithm, such as \texttt{uniroot} in \textsf{R}, treating $\Delta_{\mathcal{K}^{(2)}}$ as the scalar argument with fixed $\hat{\Delta}_{\mathcal{K}^{(2)}}$.

\section{Example}\label{sec:example}

In this section, we consider an adaptive enrichment trial and apply the proposed confidence interval procedures to demonstrate their use in a realistic setting. The Early Minimally Invasive Removal of Intracerebral Hemorrhage (ENRICH) trial \citep{pradilla2024trial} used a multi-stage adaptive design that considered potential treatment effect heterogeneity across patient subpopulations. Patients were classified according to hemorrhage location: lobar hemorrhage ($\mathcal{S}_1$) and anterior basal ganglia ($\mathcal{S}_2$). The trial planned a sample size ranging from 150 to 300 patients, with interim analyses triggered after 150, 175, 200, 225, 250, and 275 patients had been enrolled. A total of 300 patients were ultimately recruited. Following an interim analysis conducted after 175 patients, subsequent recruitment was restricted to patients with lobar hemorrhage ($\mathcal{S}_1$).  In the confirmatory analysis, treatment effects in $\mathcal{S}_1$, $\mathcal{S}_2$, and $\mathcal{F}$ were estimated. The subgroup definition, multi-stage conduct, enrollment adaptation, and co-primary analysis features align with the adaptive enrichment designs considered in this paper.

To illustrate the proposed confidence interval procedures within a setting consistent with the ENRICH trial, we consider a simplified two-stage adaptive enrichment design. We assume a maximum sample size of 300 patients, with a single interim analysis conducted after 200 patients using the decision rule defined in \cref{fig:decision2}. Patient accrual follows the framework described in \cref{sec:notation}. This simplified structure aligns with the class of designs studied in this paper. As in the ENRICH trial, the primary endpoint is the utility-weighted modified Rankin scale at 180 days. Following published studies \citep{mendelow2013early,hanley2016safety,pradilla2024trial}, we assume the endpoint follows a normal distribution with common variance across treatment arms and subpopulations, with $\sigma^2 = 0.36^2$. The subgroup prevalences are approximated by $p_1 = p_2 = 0.5$ for analytical simplicity. This approximation is consistent with the observed early accrual in the ENRICH trial, where among the first 175 enrolled patients, 83 (47.4\%) and 92 (52.6\%) were classified as lobar and anterior basal ganglia hemorrhages, respectively. The interim decision threshold is set to $\Delta^* = 0.025$. This value is chosen such that, if the true treatment effect in the full population equals the minimum clinically meaningful effect reported in the ENRICH trial, the stage 1 estimator $\hat{\Delta}^{(1)}_{\mathcal{K}}$ exceeds $\Delta^*$ with probability over $80\%$.

Stagewise summaries from the ENRICH trial are not reported. Therefore, for illustration purposes, we construct stagewise sample means that are numerically consistent with the reported overall sample means in $\mathcal{S}_1$ and $\mathcal{S}_2$. The stage 1 sample means in $\mathcal{S}_1$, $\mathcal{S}_2$, and the full population $\mathcal{F}$ are taken to be $\hat{\Delta}_{1}^{(1)} = 0.113$, $\hat{\Delta}_{2}^{(1)} = 0.013$, and $\hat{\Delta}_{\{1,2\}}^{(1)} = 0.063$, respectively. According to the decision rule, the trial proceeds with the full population. At stage 2, the sample means are taken to be $\hat{\Delta}_{1}^{(2)} = 0.155$, $\hat{\Delta}_{2}^{(2)} = -0.064$, and $\hat{\Delta}_{\{1,2\}}^{(2)} = 0.045$. The corresponding overall naive estimates are $\hat{\Delta}_{1}= 0.127$, $\hat{\Delta}_{2} = -0.013$, and $\hat{\Delta}_{\{1,2\}} = 0.057$. Using these quantities, we calculate 95\% confidence intervals for the treatment effects in $\mathcal{F}$, $\mathcal{S}_1$, and $\mathcal{S}_2$ using the naive, C-UMAU, and C-TOST approaches. Details for constructing the C-UMAU confidence intervals in this example are provided in Appendix B.2.

\begin{table}[!htbp]
	\centering
	\begin{tabular}{l c c c}
		\hline
		Method & CI for $\Delta_{\{1,2\}}$ & CI for $\Delta_{1}$ & CI for $\Delta_{2}$ \\ \hline
		Naive & (-0.024, 0.138)& (0.012, 0.242) & (-0.128, 0.102) \\
		C-UMAU & (-0.079, 0.131)& (-0.028, 0.240) & (-0.200, 0.093) \\
		C-TOST & (-0.078, 0.132)& (-0.025, 0.240) & (-0.198, 0.094) \\ \hline	
	\end{tabular}
	\caption{95\% confidence intervals for treatment effects in the worked example under the naive, C-UMAU, and C-TOST methods.}
	\label{tab:example}
\end{table}

\cref{tab:example} reports the $95\%$ confidence intervals under the naive, C-UMAU, and C-TOST approaches. Across the full population and both subgroups, the C-UMAU and C-TOST intervals are consistently wider than the naive intervals, primarily due to lower bounds being reduced while upper bounds remain similar. The naive confidence interval ignores the interim decision and treats the data as arising from a fixed design, leading to narrower intervals. In contrast, C-UMAU and C-TOST condition on the adaptive selection and therefore adjust for the adaptive decision mechanism, resulting in wider and lower intervals. In addition. the C-UMAU intervals are marginally wider than the C-TOST intervals.

\section{Simulation Study} \label{sec:sim}

To evaluate the conditional coverage of the C-UMAU confidence interval and the associated width inflation required to achieve conditional coverage, we conduct the simulation using adaptive enrichment designs with interim decision rules described in \cref{sec:designs}.

\subsection{Simulation Results for the Adaptive Enrichment Design in \citet{rosenblum2013confidence}} \label{sec:sim1}
We first consider the adaptive enrichment design described in Section 6 of \citet{rosenblum2013confidence}, which evaluated a new antidepressant. Based on the meta-analysis of \citet{kirsch2008initial}, such a treatment may benefit only patients with severe depression but not those with moderate depression. We assume $\mathcal{S}_1$ consists of patients with severe baseline depression and $\mathcal{S}_2$ of those with moderate depression, with $\mathcal{K} = \{1,2\}$. Patient enrollment follows the scheme described in \cref{sec:notation}. The outcome is the change in the Hamilton Rating Scale for Depression (HRSD) score from baseline to the final visit. Based on \citet{kirsch2008initial}, HRSD outcomes are assumed normally distributed with known standard deviation $\sigma = 8$ HRSD points. Decision rule $D_1$ described in \cref{sec:design1} is used for the interim analysis.

We evaluate the performance of the C-UMAU confidence interval, the C-TOST confidence interval, the confidence interval procedure in \citet{rosenblum2013confidence}, and the naive confidence interval. The confidence interval procedure of \citet{rosenblum2013confidence} asymptotically guarantees the overall (i.e., unconditional) coverage, which is given by
\begin{align*}
	\Pr\!\big(\Delta_{\mathcal{K}^{(2)}}  \in  \text{CI}(\Delta_{\mathcal{K}^{(2)}}) \big) = & \Pr\!\big(D_1(\mathcal{X}^{(1)}) = \{1,2\} \big) \Pr\!\big( \Delta_{\mathcal{K}} \in \text{CI}(\Delta_{\mathcal{K}}) \,\big|\, D_1(\mathcal{X}^{(1)}) = \{1,2\} \big) \\
	& + \Pr\!\big(D_1(\mathcal{X}^{(1)}) = \{1\} \big) \Pr\!\big( \Delta_{1} \in \text{CI}(\Delta_{1}) \,\big|\, D_1(\mathcal{X}^{(1)}) = \{1\} \big) \\
	& + \Pr\!\big(D_1(\mathcal{X}^{(1)}) = \{2\} \big) \Pr\!\big( \Delta_{2} \in \text{CI}(\Delta_{2}) \,\big|\, D_1(\mathcal{X}^{(1)}) = \{2\} \big).
\end{align*}
The construction of the C-UMAU confidence intervals for this design is described in Appendix B.1.

We use the same parameter setting as in \citet{rosenblum2013confidence}. Three hypothetical scenarios are considered. In scenario 1, suppose the antidepressant benefits the full population equally such that $\Delta_1 = \Delta_2 = 1.8$. In scenario 2, we assume the antidepressant only benefits subpopulation 1 with $\Delta_1 = 1.8$ and $\Delta_2 = 0$. For scenario 3, there is no benefit  in either subpopulation, that is, $\Delta_1 = \Delta_2 = 0$. As in \citet{rosenblum2013confidence}, we set $Z_*=  1$ in $D_1$. The sample size is set as $n=488$ $(n_1 = n_2 = 244)$ and the proportion of each subpopulation $p_1 = p_2 = 0.5$. In each scenario, for subpopulation $\mathcal{S}_m$ ($m \in {1,2}$), the mean outcome in the control arm is set to $\mu_{m0}=0$, and the mean in the treatment arm to $\mu_{m1} = \Delta_m$.

\begin{table}[!htbp]
	\centering
	\begin{tabular}{l c c c c c}
		\hline
		Scenario & Method & Continue with $\mathcal{F}$ & Enrich $\mathcal{S}_1 $ & Enrich to $\mathcal{S}_2 $ & Overall \\ \hline
		\multirow{5}{*}{\makecell[l]{Scenario 1: \\ $\Delta_1 = 1.8$  \\ $\Delta_2 = 1.8$ }} & Decision proportion & $77.35\%$ & $11.41\%$ & $11.25\%$ & - \\
		& Naive & $96.38\%$ & $97.09\%$ & $96.97\% $ & $96.53\%$   \\
		& \citet{rosenblum2013confidence} & $96.88\% \ (1.03)$ & $97.50\% \ (1.03)$ & $97.45\% \ (1.03) $ & $97.02\% \ (1.03) $ \\
		& C-UMAU & $\textbf{95.01\%} \ (1.14)$ & $\textbf{95.06\%} \ (1.16)$ & $\textbf{95.04\%} \ (1.16)$ & $\textbf{95.02\%} \ (1.15)$   \\
		& C-TOST & $\textbf{94.99\%} \ (1.14)$ & $\textbf{95.03\%} \ (1.16)$ & $\textbf{95.03\%} \ (1.16)$ & $\textbf{95.00\%} \ (1.15)$   \\ \hline
		\multirow{5}{*}{\makecell[l]{Scenario 2: \\ $\Delta_1 = 1.8$  \\ $\Delta_2 = 0$ }} & Decision proportion & $45.16\%$ & $44.43\%$ & $10.41\%$ & - \\
		& Naive & $94.36\% $ & $96.66\% $ & $96.60\%$ & $95.62\% $   \\
		& \citet{rosenblum2013confidence} & $\textbf{95.11\%} \ (1.03) $ & $97.16\% \ (1.03) $ & $97.03\% \ (1.03) $ & $96.22\% \ (1.03) $ \\
		& C-UMAU & $\textbf{95.13\%}  \ (1.21) $ & $\textbf{95.02\%} \ (1.12) $ & $\textbf{94.76\%} \ (1.16) $ &$\textbf{95.05\%} \ (1.16) $ \\
		& C-TOST & $\textbf{95.10\%}  \ (1.21) $ & $\textbf{95.02\%} \ (1.11) $ & $\textbf{94.77\%} \ (1.16) $ &$\textbf{95.03\%} \ (1.16) $ \\ \hline
		\multirow{5}{*}{\makecell[l]{Scenario 3: \\ $\Delta_1 = 0$  \\ $\Delta_2 = 0$ }} & Decision proportion & $15.90\%$ & $42.00\%$ & $42.10\%$ & - \\
		& Naive & $87.76\% $ & $96.37\%$ & $96.26\% $ & $\textbf{94.95\%} $   \\
		& \citet{rosenblum2013confidence} & $89.07\% \ (1.03) $ & $96.87\% \ (1.03) $ & $96.84\% \ (1.03) $ & $95.62\% \ (1.03) $ \\
		& C-UMAU & $\textbf{95.03\%} \ (1.27)$ & $\textbf{94.95\%} \ (1.12)$ & $ \textbf{94.96\%} \ (1.12)$  & $\textbf{94.97\%} \ (1.14)$ \\
		& C-TOST & $\textbf{95.04\%} \ (1.28)$ & $\textbf{94.96\%} \ (1.12)$ & $ \textbf{94.96\%} \ (1.12)$  & $\textbf{94.97\%} \ (1.14)$ \\ \hline
	\end{tabular}
	\caption{Simulated interim decision proportions, coverage probabilities, and relative average interval widths under three scenarios for the design described in \cref{sec:sim1}. Boldface coverage values indicate that the empirical coverage lies within the 95\% Monte Carlo error band around the nominal 95\% level. Values in parentheses denote the ratio of the average interval width to that of the naive confidence interval.}
	\label{tab:width}
\end{table}

Under each scenario, we conduct 100,000 simulated trials and compute the corresponding $95\%$ confidence intervals and their widths. The proportions of interim decisions, the conditional coverage probabilities of the considered confidence interval methods given each interim decision, as well as the ratios of their average widths relative to the naive confidence interval, are summarized in \cref{tab:width}. As described in \cref{sec:decision}, the C-UMAU guarantees a conditional coverage probability of $95\%$, which also ensures the overall coverage to be $95\%$ in this design. The simulation results confirm this property, with both the conditional and overall coverage probabilities of the C-UMAU confidence intervals lying within the $95\%$ Monte Carlo error band around the nominal $95\%$ level under all three scenarios. The average width ratio between the C-UMAU and naive intervals ranges from 1.12 to 1.27, indicating that an additional $12\%-27\%$ increase in width is required to guarantee conditional $95\%$ coverage. For comparison, the C-TOST confidence interval also attains conditional and overall coverage close to the nominal $95\%$ level, with a degree of width inflation similar to that observed for the C-UMAU confidence interval. By construction, the confidence interval of \citet{rosenblum2013confidence} guarantees overall coverage of at least $95\%$ asymptotically, but not conditional coverage for each interim decision. The average width ratio between this interval and the naive one remained fixed at 1.03. Notably, the overall coverage probabilities of the naive confidence interval were also close to or slightly above $95\%$ under the three scenarios. Thus, in terms of overall coverage, the improvement offered by the interval of \citet{rosenblum2013confidence} over the naive confidence interval is limited in this simulation study.

For the proportions of interim decisions, in Scenario 1 the design continued with the full population in a substantial proportion of trials $(77.35\%)$. When the treatment benefited only subpopulation 1, as in Scenario 2, the design enriched to subpopulation 1 in $44.43\%$ of trials. Without a stopping for futility option, an inherent limitation of $D_1$ is that the trial always continues in some form, regardless of how unfavorable the stage 1 results are, as illustrated in Scenario 3.

\subsection{Simulation Results for an Adaptive Enrichment Design with Stopping for futility and Co-primary Analysis} \label{sec:sim2}
We modify the adaptive enrichment design proposed in \citet{rosenblum2013confidence} by using interim decision rules $D_2$ described in \cref{sec:design2}, where a stopping for futility is incorporated. We also consider the co-primary analysis as described in \cref{sec:co-primary} such that the confidence intervals for treatment effects in $\mathcal{S}_1$ and $\mathcal{S}_2$ are calculated when the trial continues with $\mathcal{F}$. 

We consider the same settings and scenarios as in the previous section. We set $\Delta_* = 1$ in $D_2$, which yields reasonable decisions, as we show below. Extending the method of \citet{rosenblum2013confidence} to this design is not straightforward; therefore, we evaluate the performance of the proposed C-UMAU confidence interval, the C-TOST confidence interval, and the naive confidence interval. Additional details on constructing the C-UMAU confidence intervals for this design are given in Appendix B.2. We report the same metrics as in the previous section in \cref{tab:coverage}. The last column of \cref{tab:coverage} reports the overall coverage probability and the average width ratio for all simulated trials that continue to stage 2, regardless of the selected population.

Under all three scenarios, both the conditional and overall coverage probabilities of the C-UMAU confidence intervals fall within the $95\%$ Monte Carlo error band around the nominal $95\%$ coverage level. As in the previous simulation setting, the C-TOST confidence interval shows similar empirical performance with respect to coverage and interval width. The naive confidence interval generally attains at least $95\%$ coverage when the true treatment effect is substantial, but may fail to reach the nominal level when the effect is weak or absent as in Scenario 3. Given each interim decision, the width inflation of the C-UMAU intervals relative to the naive intervals ranges from $14\%$ to $27\%$, and the overall width inflation under the three scenarios lies between $15\%$ and $22\%$.

In Scenario 1, the design continued with the full population in a large proportion of trials $(78.29\%)$. In Scenario 2, an appreciable fraction $(28.00\%)$ of trials enriched to subpopulation 1, where the treatment effect was more pronounced. In Scenario 3, where neither subpopulation benefited, $57.26\%$ of trials stopped for futility at the interim. Compared with $D_1$, the interim decision rule $D_2$ can effectively terminate the trial when treatment effects in both subpopulations are absent.

For co-primary analysis, when the trial continues with the full population, we also evaluate the C-UMAU and the C-TOST confidence intervals for treatment effects in both $\mathcal{S}_1$ and $\mathcal{S}_2$. The simulation results in \cref{tab:co-primary} demonstrate that the conditional coverage of the C-UMAU and the C-TOST confidence intervals for $\Delta_{1}$ and $\Delta_{2}$ is maintained at the nominal $95\%$ level, with all simulated coverage probabilities lying within simulation error of $95\%$. In contrast, the naive confidence interval fails to achieve the nominal level in scenario 3 where the treatment effect is absent. The width inflation of the C-UMAU and C-TOST intervals in this co-primary analysis ranges from $10\%$ to $23\%$.

\clearpage
\begin{sidewaystable}
	\centering
	\begin{tabular}{l c c c c c c}
		\hline
		Scenario & Method & Continue with $\mathcal{F}$ & Enrich $\mathcal{S}_1 $ & Enrich to $\mathcal{S}_2 $ & Stop &Continuation \\ \hline
		\multirow{5}{*}{\makecell[l]{Scenario 1: \\ $\Delta_1 = 1.8$  \\ $\Delta_2 = 1.8$ }} & Decision proportion & $78.29\%$ & $6.70\%$ & $6.66\%$ & $4.09\%$ & $91.65\%$ \\
		& Naive & $96.42\%$ & $98.10\%$ & $97.99\% $ & $-$ & $96.66\%$   \\
		& C-UMAU & $\textbf{95.02\%} \ (1.14)$ & $\textbf{94.99\%} \ (1.20)$ & $\textbf{95.23\%} \ (1.20)$ & $-$ & $\textbf{95.03\%} \ (1.15)$   \\
		& C-TOST & $\textbf{95.01\%} \ (1.14)$ & $\textbf{94.99\%} \ (1.20)$ & $\textbf{95.23\%} \ (1.20)$ & $-$ & $\textbf{95.03\%} \ (1.15)$   \\ \hline
		\multirow{5}{*}{\makecell[l]{Scenario 2: \\ $\Delta_1 = 1.8$  \\ $\Delta_2 = 0$ }} & Decision proportion & $46.46\%$ & $28.00\%$ & $3.73\%$ & $21.81\%$ & $78.19\%$ \\
		& Naive & $94.69\%$ & $97.54\%$ & $94.17\% $ & $-$ & $95.69\%$   \\
		& C-UMAU & $\textbf{94.97\%} \ (1.21)$ & $\textbf{94.90\%} \ (1.18)$ & $\textbf{95.06\%} \ (1.21)$ & $-$ & $\textbf{94.95\%} \ (1.20)$   \\
		& C-TOST & $\textbf{94.97\%} \ (1.21)$ & $\textbf{94.91\%} \ (1.18)$ & $\textbf{95.06\%} \ (1.20)$ & $-$ & $\textbf{94.96\%} \ (1.20)$   \\ \hline
		\multirow{5}{*}{\makecell[l]{Scenario 3: \\ $\Delta_1 = 0$  \\ $\Delta_2 = 0$ }} & Decision proportion & $16.34\%$ & $13.07\%$ & $13.33\%$ & $57.26\%$ & $42.74\%$ \\
		& Naive & $87.06\%$ & $93.25\%$ & $93.92\% $ & $-$ & $91.10\%$   \\
		& C-UMAU & $\textbf{94.86\%} \ (1.27)$ & $\textbf{94.67\%} \ (1.19)$ & $\textbf{95.22\%} \ (1.19)$ & $-$ & $\textbf{94.92\%} \ (1.22)$   \\
		& C-TOST & $\textbf{94.90\%} \ (1.27)$ & $\textbf{94.68\%} \ (1.19)$ & $\textbf{95.21\%} \ (1.19)$ & $-$ & $\textbf{94.98\%} \ (1.22)$   \\ \hline
	\end{tabular}
	\caption{Simulated interim decision proportions, coverage probabilities, and relative average interval widths under three scenarios for the design described in \cref{sec:sim2}. Boldface coverage values indicate that the empirical coverage lies within the 95\% Monte Carlo error band around the nominal 95\% level. Values in parentheses denote the ratio of the average interval width to that of the naive confidence interval.}
	\label{tab:coverage}
\end{sidewaystable} 

\clearpage

\begin{table}[!htbp]
	\centering
	\begin{tabular}{c c c c}
		\hline
		Scenario & Method & CI for $\Delta_1$ in the co-primary option & CI for $\Delta_2$ in the co-primary option \\ \hline
		Scenario 1: & Naive & $95.68\%$ & $95.79\%$  \\
		$\Delta_1 = 0.5$ & C-UMAU & $\textbf{94.92\%} \ (1.10)$ & $\textbf{95.01\%} \ (1.10)$ \\
		$\Delta_2 = 0.5$ & C-TOST & $\textbf{94.91\%} \ (1.10)$ & $\textbf{95.02\%} \ (1.10)$ \\ \hline
		Scenario 2: & Naive & $94.80\%$ & $\textbf{94.81\%}$ \\
		$\Delta_1 = 0.5$ & C-UMAU & $\textbf{94.97\%}  \ (1.17)$ & $\textbf{94.96\%}  \ (1.17)$   \\
		$\Delta_2 = 0.2$ & C-TOST & $\textbf{95.00\%}  \ (1.17)$ & $\textbf{94.98\%}  \ (1.17)$ \\ \hline
		Scenario 3: & Naive & $90.59\% $ & $90.92\% $  \\
		$\Delta_1 = 0.5$ & C-UMAU & $\textbf{94.75\%} \ (1.22) $ & $\textbf{94.91\%} \ (1.22)$   \\
		$\Delta_2 = 0$ & C-TOST & $\textbf{94.71\%} \ (1.23) $ & $\textbf{94.92\%} \ (1.22)$   \\ \hline
	\end{tabular}
	\caption{Simulated coverage probabilities and relative interval widths in the co-primary analysis under three scenarios for the design described in \cref{sec:sim2}. Boldface coverage values indicate that the empirical coverage lies within the 95\% Monte Carlo error band around the nominal 95\% level. Values in parentheses denote the ratio of the average interval width to that of the naive confidence interval.}
	\label{tab:co-primary}
\end{table}

\section{Discussion} \label{sec:discuss}
We have developed new methods for constructing confidence intervals that guarantees $100(1-\alpha)\%$ conditional coverage given each possible interim decision, for a broad class of two-stage adaptive enrichment designs which allow stopping for futility and co-primary analysis. Our simulation study confirms that the proposed confidence intervals achieve the nominal conditional coverage in finite samples and quantifies the increase in interval width required to attain this stronger form of coverage control.

For clarity of exposition, the class $\mathcal{D}$ is defined in terms of constraints on the stage 1 estimator. The proposed conditional inference framework, however, also applies to certain decision rules involving auxiliary stage 1 quantities, provided that the resulting selection event can be explicitly characterised. Details are provided in Appendix C.

The C-UMAU confidence interval is our primary proposal, as it uniquely achieves exact conditional coverage together with the uniformly most accurate unbiased property. For completeness, we also consider a simpler alternative based on inverting two conditional one-sided tests, referred to as the C-TOST confidence interval. While the C-TOST confidence interval has correct conditional coverage and is computationally simpler, it sacrifices the uniformly most accurate unbiased property and is therefore theoretically suboptimal. In simulation studies, confidence intervals are typically evaluated in terms of coverage probability and average length. These summary measures do not directly reflect the uniformly most accurate and unbiased properties that define the optimality of the C-UMAU confidence interval. As a result, the C-UMAU and C-TOST confidence intervals may exhibit similar empirical performance with respect to coverage and length, despite the stronger theoretical guarantees of the C-UMAU construction. In settings where computational simplicity is a primary concern and optimality is not critical, the C-TOST confidence interval may provide a useful alternative.

The current methodology assumes normally distributed outcomes. For non-normal outcomes, the theoretical results hold asymptotically, provided that standard regularity conditions ensure the asymptotic normality of the treatment effect estimator as the stage-wise sample sizes tend to infinity (e.g. by the central limit theorem). Extending the method to adaptive designs with more than two stages is a natural direction for future research, but remains technically challenging and is left as an open problem.

In this paper, we have focused on individual coverage for a single parameter rather than simultaneous coverage for a set of parameters. For unconditional \(100(1-\alpha)\%\) simultaneous coverage in adaptive designs, one may refer to \citet{magirr2013simultaneous}. If conditional simultaneous coverage is desired, our method could in principle be adapted by applying multiplicity corrections, such as the Bonferroni or \v{S}id\'ak adjustments \citep{bonferroni1936teoria, hochberg1987multiple}. However, the resulting intervals would no longer be uniformly most accurate unbiased and may become conservative. We also restrict attention to conditional coverage. A detailed comparison of conditional and unconditional inference lies beyond the scope of this paper; see \citet{marschner2021general} for a comprehensive discussion.

\section*{Acknowledgments}
This work is funded by the MRC Doctoral Training Partnership in Interdisciplinary Biomedical Research awarded to the University of Warwick (MR/W007053/1) and Novartis.

\newpage
\bibliographystyle{apalike}
\bibliography{C-UMAU-CI}

\appendix

\newpage
\section*{Appendix}
\subsection*{A: Theorems and Proofs} 
\subsubsection*{A.1. Proof of Theorem 1 from Section 4.1} \label{A1}
\proof
For Theorem 1(i), we first consider the conditional distribution of $(\hat{\Delta}^{(1)}, \hat{\Delta}^{(2)})$ given $l<\hat{\Delta}^{(1)}<u$, which is 
\begin{equation*}
	f(\hat{\Delta}^{(1)}, \hat{\Delta}^{(2)} \mid l<\hat{\Delta}^{(1)}<u) = \frac{1}{\sigma_{(1)}}\phi (\frac{\hat{\Delta}^{(1)} - \Delta}{\sigma_{(1)}}) \frac{1}{\Phi(\frac{u - \Delta}{\sigma_{(1)}}) - \Phi(\frac{l - \Delta}{\sigma_{(1)}})}  \frac{1}{\sigma_{(2)}}\phi (\frac{\hat{\Delta}^{(2)} - \Delta}{\sigma_{(2)}}).
\end{equation*}
Part of the expression above can be written as
\begin{align*}
	& \phi (\frac{\hat{\Delta}^{(1)} - \Delta}{\sigma_{(1)}}) \phi (\frac{\hat{\Delta}^{(2)} - \Delta}{\sigma_{(2)}}) \\
	= & \phi(\frac{ \frac{\tau_{(1)}}{\tau_{(1)} + \tau_{(2)}}\hat{\Delta}^{(1)} + \frac{\tau_{(2)}}{\tau_{(1)} + \tau_{(2)}}\hat{\Delta}^{(2)} - \Delta}{\sigma_{(1)} \sigma_{(2)}/ \sqrt{\sigma_{(1)}^2 + \sigma_{(2)}^2}})  \phi(\frac{\hat{\Delta}^{(1)} - (\frac{\tau_{(1)}}{\tau_{(1)} + \tau_{(2)}}\hat{\Delta}^{(1)} + \frac{\tau_{(2)}}{\tau_{(1)} + \tau_{(2)}}\hat{\Delta}^{(2)})}{\sigma_{(1)}^2/ \sqrt{\sigma_{(1)}^2 + \sigma_{(2)}^2}})
\end{align*}
Let $\hat{\Delta}=\tau_{(1)}/(\tau_{(1)} + \tau_{(2)})\hat{\Delta}^{(1)} + \tau_{(2)}/(\tau_{(1)} + \tau_{(2)})\hat{\Delta}^{(2)}$ and $\sigma_{(12)}^2 = \sigma_{(1)}^2 \sigma_{(2)}^2/ (\sigma_{(1)}^2 + \sigma_{(2)}^2)$. We can obtain the conditional distribution of $(\hat{\Delta}^{(1)}, \hat{\Delta})$ given $l<\hat{\Delta}^{(1)}<u$ by transforming $(\hat{\Delta}^{(1)}, \hat{\Delta}^{(2)})$ to $(\hat{\Delta}^{(1)}, \hat{\Delta})$ as
\begin{equation*}
	f(\hat{\Delta}^{(1)}, \hat{\Delta} \mid l<\hat{\Delta}^{(1)}<u) = \frac{1}{\sigma_{(12)}} \phi(\frac{\hat{\Delta} - \Delta}{\sigma_{(12)}})  \frac{1}{(\sigma_{(1)}/\sigma_{(2)})\sigma_{(12)}}  \phi(\frac{\hat{\Delta}^{(1)}- \hat{\Delta}}{(\sigma_{(1)}/\sigma_{(2)})\sigma_{(12)}})   \frac{1}{\Phi(\frac{u - \Delta}{\sigma_{(1)}}) - \Phi(\frac{l - \Delta}{\sigma_{(1)}})}.
\end{equation*}
The conditional distribution of $\hat{\Delta}$ given $l<\hat{\Delta}^{(1)}<u$ can be obtained by integrating out $\hat{\Delta}^{(1)}$ in the above equation as
\begin{align*}
	f(\hat{\Delta} \mid l<\hat{\Delta}^{(1)}<u) & = \int_{l}^{u} f(\hat{\Delta}^{(1)}, \hat{\Delta} \mid l<\hat{\Delta}^{(1)}<u) \ \mathrm{d} \hat{\Delta}^{(1)} \\
	& = \frac{1}{\sigma_{(12)}} \phi(\frac{\hat{\Delta} - \Delta}{\sigma_{(12)}}) \frac{\Phi(\frac{u - \hat{\Delta}}{(\sigma_{(1)}/\sigma_{(2)}) \sigma_{(12)}}) - \Phi(\frac{l - \hat{\Delta}}{(\sigma_{(1)}/\sigma_{(2)}) \sigma_{(12)}})}{\Phi(\frac{u - \Delta}{\sigma_{(1)}}) - \Phi(\frac{l - \Delta}{\sigma_{(1)}})}.
\end{align*}
The conditional distribution of $\hat{\Delta}$ belongs to the one-parameter exponential family. By the factorization theorem, $\hat{\Delta}$ is a sufficient statistic for $\Delta$ within this conditional model.

For Theorem 1(ii), the conditional distribution $f(\hat{\Delta} \mid l<\hat{\Delta}^{(1)}<u)$ follows the one-parameter exponential family, the construction of the two-sided uniformly most powerful unbiased test for the one-parameter exponential family is given in Section 4.2 in \citet{lehmann1986testing}.

For Theorem 1(iii), it is straightforward to verify $f(\hat{\Delta} \mid l<\hat{\Delta}^{(1)}<u)$ has monotone likelihood ratio in $\hat{\Delta}$. The expression of the uniformly most accurate unbiased confidence interval is given in Lemma 5.5.1 in \citet{lehmann1986testing}.
\endproof

\subsubsection*{A.2. Proof of Proposition 1 from Section 4.1} \label{A3}
\begin{proof}
	For Proposition 1(i), since $C_1$ and $C_2$ are strictly increasing in $\Delta_0$ (see Lemma 5.5.1 in \citet{lehmann1986testing}), we have
	\begin{align*}
		& \Pr \big(\Delta \in (\underline{\Delta},\ \overline{\Delta}) \mid l < \hat{\Delta}^{(1)} < u \big) \\ = & \Pr \big( C_1(\Delta) < \hat{\Delta} < C_2(\Delta) \mid l < \hat{\Delta}^{(1)} < u \big) = 1- \mathbb{E}_{\Delta}(\psi_{\Delta}( \hat{\Delta})) = 1 - \alpha.
	\end{align*}
	Proposition 1(ii) and Proposition 1(iii) follow the definition of the uniformly most accurate unbiased confidence interval in Section 5.5 in \citet{lehmann1986testing}.
\end{proof}

\subsubsection*{A.3. Corollary 1} \label{A5}

\begin{corollary}\label{corollary}
	Suppose $\hat{\Delta}$ is distributed as in Theorem 1(i). Then, its expectation is given by
	\begin{equation*}
		\mathbb{E}_{\Delta}( \hat{\Delta}) = \Delta + \frac{\phi(\frac{l-\Delta}{\sigma_{(1)}}) - \phi(\frac{u-\Delta}{\sigma_{(1)}})}{\Phi(\frac{u-\Delta}{\sigma_{(1)}}) - \Phi(\frac{l-\Delta}{\sigma_{(1)}})} \frac{\sigma_{(1)} \tau_{(1)}}{\tau_{(1)} + \tau_{(2)}}.
	\end{equation*} 
\end{corollary}

\begin{proof}
	\begin{align*}
		\mathbb{E}_{\Delta}( \hat{\Delta}) & = \mathbb{E}_{\Delta}(\frac{\tau_{(1)}}{\tau_{(1)} + \tau_{(2)}}\hat{\Delta}^{(1)} + \frac{\tau_{(2)}}{\tau_{(1)} + \tau_{(2)}}\hat{\Delta}^{(2)}) =  \frac{\tau_{(1)}}{\tau_{(1)} + \tau_{(2)}} \mathbb{E}_{\Delta}(\hat{\Delta}^{(1)}) + \frac{\tau_{(2)}}{\tau_{(1)} + \tau_{(2)}} \mathbb{E}_{\Delta}(\hat{\Delta}^{(2)}),
	\end{align*} 
	where $\hat{\Delta}^{(1)}$ is truncated normal distribution with mean $\Delta$ and standard deviation $\sigma_{(1)}$ and lies within the interval $(l,u)$, and $\hat{\Delta}^{(2)}$ is truncated normal distribution with mean $\Delta$ and standard deviation $\sigma_{(2)}$.
	Since
	\begin{equation*}
		\mathbb{E}_{\Delta}(\hat{\Delta}^{(1)}) = \Delta + \frac{\phi(\frac{l-\Delta}{\sigma_{(1)}}) - \phi(\frac{u-\Delta}{\sigma_{(1)}})}{\Phi(\frac{u-\Delta}{\sigma_{(1)}}) - \Phi(\frac{l-\Delta}{\sigma_{(1)}})} \sigma_{(1)}, \quad \mathbb{E}_{\Delta}(\hat{\Delta}^{(2)}) = \Delta, 
	\end{equation*}
	we have
	\begin{equation*}
		\mathbb{E}_{\Delta}( \hat{\Delta}) = \Delta + \frac{\phi(\frac{l-\Delta}{\sigma_{(1)}}) - \phi(\frac{u-\Delta}{\sigma_{(1)}})}{\Phi(\frac{u-\Delta}{\sigma_{(1)}}) - \Phi(\frac{l-\Delta}{\sigma_{(1)}})} \frac{\sigma_{(1)} \tau_{(1)}}{\tau_{(1)} + \tau_{(2)}}.
	\end{equation*} 
\end{proof}

\subsubsection*{A.4. Lemma 1} \label{A4}

\begin{lemma}\label{lemma}
	Let $f(\cdot)$ and $F(\cdot)$ denote the probability density function and cumulative distribution function of the distribution in Theorem 1(i). 
	For any $c_1$ such that $F(c_1) < \alpha$ $(0<\alpha<1/2)$, define 
	\begin{equation*}
		c_2(c_1) = F^{-1}\bigl(F(c_1) + 1 - \alpha\bigr).
	\end{equation*}
	Then the function
	\begin{equation*}
		I(c_1) = \int_{c_1}^{c_2(c_1)} t f(t)\, dt
	\end{equation*}
	is continuous and strictly increasing in $c_1$.
\end{lemma}

\begin{proof}
	We first prove the continuity of $I(c_1)$, It is straightforward to show that $c_2(c_1)$ is continuous in $c_1$. In addition, it follows directly from its expression that $\hat{\Delta} f(\hat{\Delta} \mid l<\hat{\Delta}^{(1)}<u)$ is bounded. Let $M = \sup_{\hat{\Delta}} \hat{\Delta} f(\hat{\Delta} \mid l<\hat{\Delta}^{(1)}<u)$. Let $\delta_1 = c_2(c_1) - c_1$. Then, we have $\forall c_1': |c_1' - c_1|<\delta_1$ satisfies $c_1' < c_2(c_1)$. By the continuity of $c_2(c_1)$ , $\exists \delta_2 \;\text{s.t.} \; \forall c_1': |c_1' - c_1|<\delta_2$ satisfies $c_2(c_1') > c_1$. For any $\epsilon>0$, let $\delta_3 = \epsilon/(2M)$. By the continuity of $c_2(c_1)$, $\exists \delta_4 \;\text{s.t.} \; \forall c_1': |c_1' - c_1|<\delta_4$ has $|c_2(c_1') - c_2(c_1)| < \epsilon/(2M)$. Let $\delta = \min(\delta_1, \delta_2, \delta_3, \delta_4)$. For $\forall c_1': |c_1' - c_1| < \delta$, we have
	\begin{align*}
		& |\int_{c_1'}^{c_2(c_1')} \hat{\Delta} f(\hat{\Delta} \mid l<\hat{\Delta}^{(1)}<u) \mathrm{d}\hat{\Delta} - \int_{c_1}^{c_2(c_1)}  \hat{\Delta} f(\hat{\Delta} \mid l<\hat{\Delta}^{(1)}<u) \mathrm{d}\hat{\Delta}| \\ < &|\int_{\min(c_1,c_1' )}^{\max(c_1,c_1')}  \hat{\Delta} f(\hat{\Delta} \mid l<\hat{\Delta}^{(1)}<u) \mathrm{d}\hat{\Delta}| + |\int_{\min(c_2(c_1),c_2(c_1'))}^{\max(c_2(c_1),c_2(c_1' ))}  \hat{\Delta} f(\hat{\Delta} \mid l<\hat{\Delta}^{(1)}<u) \mathrm{d}\hat{\Delta} dt| \\ < & |c_1' - c_1|M + |c_2(c_1') - c_2(c_1)|M \\ < & \epsilon
	\end{align*}
	Then, we prove the monotonicity. Suppose $c_1 < c_1'$. If $c_1' \geq c_2(c_1)$, we have
	\begin{equation*}
		\int_{c_1'}^{c_2(c_1')} \hat{\Delta} f(\hat{\Delta} \mid l<\hat{\Delta}^{(1)}<u) \mathrm{d}\hat{\Delta} > c_1' (1 - \alpha) \geq c_2(c_1)(1 - \alpha) > \int_{c_1}^{c_2(c_1)} \hat{\Delta} f(\hat{\Delta} \mid l<\hat{\Delta}^{(1)}<u) \mathrm{d}\hat{\Delta}.
	\end{equation*}
	Otherwise,
	\begin{align*}
		& \int_{c_1'}^{c_2(c_1')} \hat{\Delta} f(\hat{\Delta} \mid l<\hat{\Delta}^{(1)}<u) \mathrm{d}\hat{\Delta} - \int_{c_1}^{c_2(c_1)} \hat{\Delta} f(\hat{\Delta} \mid l<\hat{\Delta}^{(1)}<u) \mathrm{d}\hat{\Delta} \\
		= & \int_{c_2(c_1)}^{c_2(c_1')} \hat{\Delta} f(\hat{\Delta} \mid l<\hat{\Delta}^{(1)}<u) \mathrm{d}\hat{\Delta} - \int_{c_1}^{c_1'} \hat{\Delta} f(\hat{\Delta} \mid l<\hat{\Delta}^{(1)}<u) \mathrm{d}\hat{\Delta} \\
		> & c_2(c_1) \int_{c_2(c_1)}^{c_2(c_1')} f(\hat{\Delta} \mid l<\hat{\Delta}^{(1)}<u) \mathrm{d}\hat{\Delta}  - c_1' \int_{c_1}^{c_1'} f(\hat{\Delta} \mid l<\hat{\Delta}^{(1)}<u) \mathrm{d}\hat{\Delta}.
	\end{align*}
	Since 
	\begin{equation*}
		\int_{c_2(c_1)}^{c_2(c_1')} f(\hat{\Delta} \mid l<\hat{\Delta}^{(1)}<u) \mathrm{d}\hat{\Delta}  = \int_{c_1}^{c_1'} f(\hat{\Delta} \mid l<\hat{\Delta}^{(1)}<u) \mathrm{d}\hat{\Delta},
	\end{equation*}
	the above result is greater than 0.
\end{proof}

\subsubsection*{A.5. Proposition 2} \label{A2}
\setcounter{proposition}{1}

\begin{proposition} \label{proposition:1}
	$C_1(\Delta_0)$ and $C_2(\Delta_0)$ are continuous and strictly increasing in $\Delta_0$.
\end{proposition}

\begin{proof}
	Lemma 5.5.1 in \citet{lehmann1986testing} gives that $C_1$ and $C_2$ are strictly increasing functions in $\Delta_0$. We prove the continuity by contradiction. Suppose any of $C_i$, say $C_1$, is not continuous at $\Delta_0 = \Delta'$. Since $C_1$ is strictly increasing, $C_1$ can only be the jump discontinuity at $\Delta'$, where $\lim_{\Delta_0 \to \Delta'^-}C_1(\Delta_0) = C_1^- (\Delta')$, $\lim_{\Delta_0 \to \Delta'^+}C_1(\Delta_0) = C_1^+ (\Delta')$, and $C_1^- (\Delta') < C_1^+ (\Delta')$. Suppose $[C_1(\Delta'), C_2(\Delta')]$ is the acceptance region of the uniformly most powerful unbiased test of $H: \Delta = \Delta'$. At least one of the quantities $C_1^- (\Delta')$ or $C_1^+ (\Delta')$ differs from $C_1 (\Delta')$, say $C_1^+ (\Delta') \neq  C_1(\Delta')$. Define sequence $\Delta_n = \Delta' + 1/n \ (n \in \mathbb{Z}^+)$. Suppose $[C_1(\Delta_n), C_2(\Delta_n)]$ are the corresponding acceptance region of the uniformly most powerful unbiased test of $H: \Delta = \Delta_n$. For any $\Delta_n$, by the definition of the uniformly most powerful unbiased test, we have
	\begin{align*}
		\int_{C_1(\Delta_n)}^{C_2(\Delta_n)} f_{\Delta_n}(\hat{\Delta}  \mid l<\hat{\Delta}^{(1)}<u)  \mathrm{d}\hat{\Delta} = 1 - \alpha \intertext{and} \int_{C_1(\Delta_n)}^{C_2(\Delta_n)} \hat{\Delta} f_{\Delta_n}(\hat{\Delta}  \mid l<\hat{\Delta}^{(1)}<u)  \mathrm{d}\hat{\Delta} = (1 - \alpha)\mathbb{E}_{\Delta_n}(\hat{\Delta}).
	\end{align*}
	We have $\lim_{n \to +\infty} C_1(\Delta_n) =  C_1^+ (\Delta')$. Let $\lim_{\Delta_0 \to \Delta'^+}C_2(\Delta_0) = C_2^+ (\Delta')$. Then, we have $\lim_{n \to +\infty} C_2(\Delta_n) = C_2^+ (\Delta')$. Let $n \to +\infty$ in the above two equations, by dominated convergence theorem (details are given in S.1.2*), we obtain
	\begin{align*}
		\int_{C_1^+(\Delta')}^{C_2^+(\Delta')} f_{\Delta'}(\hat{\Delta}  \mid l<\hat{\Delta}^{(1)}<u)  \mathrm{d}\hat{\Delta} = 1 - \alpha \intertext{and} \int_{C_1^+(\Delta')}^{C_2^+(\Delta')} \hat{\Delta} f_{\Delta'}(\hat{\Delta}  \mid l<\hat{\Delta}^{(1)}<u)  \mathrm{d}\hat{\Delta} = (1 - \alpha)\mathbb{E}_{\Delta'}(\hat{\Delta}).
	\end{align*}
	Hence, $[C_1^+(\Delta'), C_2^+(\Delta')]$, where $C_1^+(\Delta') \neq C_1(\Delta')$, also constitutes the uniformly most powerful unbiased test of $H: \Delta = \Delta'$, which is conflict with the uniqueness of the uniformly most powerful unbiased test (For uniqueness, see Problem 4.23 in \citet{lehmann1986testing}).
\end{proof}

\subsubsection*{A.5*. Detail of the application of dominated convergence theorem in A.5} \label{A2*}
\begin{proof}
	We aim to prove that
	\begin{equation*}
		\lim_{n \to +\infty} \int_{C_1(\Delta_n)}^{C_2(\Delta_n)} f_{\Delta_n}(\hat{\Delta}  \mid l<\hat{\Delta}^{(1)}<u) \mathrm{d}\hat{\Delta}  = \int_{C_1^+(\Delta')}^{C_2^+(\Delta')} f_{\Delta'}(\hat{\Delta}  \mid l<\hat{\Delta}^{(1)}<u) \mathrm{d}\hat{\Delta} 
	\end{equation*}
	and 
	\begin{equation*}
		\lim_{n \to +\infty} \int_{C_1(\Delta_n)}^{C_2(\Delta_n)} \hat{\Delta} f_{\Delta_n}(\hat{\Delta}  \mid l<\hat{\Delta}^{(1)}<u) \mathrm{d}\hat{\Delta}  =  \int_{C_1^+(\Delta')}^{C_2^+(\Delta')} \hat{\Delta} f_{\Delta'}(\hat{\Delta}  \mid l<\hat{\Delta}^{(1)}<u) \mathrm{d}\hat{\Delta} .
	\end{equation*}
	Let $M = \sup_{\Delta \in [\Delta', \Delta'+1], \hat{\Delta} \in [C_1^+(\Delta'), C_2^+(\Delta')]} f_{\Delta}(\hat{\Delta}  \mid l<\hat{\Delta}^{(1)}<u)$ and $g \equiv M$. Then, in region $[C_1^+(\Delta'), C_2^+(\Delta')]$, $f_{\Delta_n}$ is dominated by $g$. $g$ is integrable on $[C_1^+(\Delta'), C_2^+(\Delta')]$ since 
	\begin{equation*}
		\int_{C_1^+(\Delta')}^{C_2^+(\Delta')} g = (C_2^+(\Delta') - C_1^+(\Delta'))M < +\infty.
	\end{equation*}
	Then, by dominated convergence theorem, we have
	\begin{equation*}
		\lim_{n \to +\infty} \int_{C_1^+(\Delta')}^{C_2^+(\Delta')} f_{\Delta_n}(\hat{\Delta}  \mid l<\hat{\Delta}^{(1)}<u) \mathrm{d}\hat{\Delta}  = \int_{C_1^+(\Delta')}^{C_2^+(\Delta')} f_{\Delta'}(\hat{\Delta}  \mid l<\hat{\Delta}^{(1)}<u) \mathrm{d}\hat{\Delta}.
	\end{equation*}
	Since $\lim_{n \to +\infty} C_1(\Delta_n) =  C_1^+ (\Delta')$ and $\lim_{n \to +\infty} C_2(\Delta_n) = C_2^+ (\Delta')$, we have
	\begin{align*}
		& \lim_{n \to +\infty} \big( \int_{C_1(\Delta_n)}^{C_2(\Delta_n)} f_{\Delta_n}(\hat{\Delta}  \mid l<\hat{\Delta}^{(1)}<u) \mathrm{d}\hat{\Delta} - \int_{C_1^+(\Delta')}^{C_2^+(\Delta')} f_{\Delta_n}(\hat{\Delta}  \mid l<\hat{\Delta}^{(1)}<u) \mathrm{d}\hat{\Delta} \big) \\
		= & \lim_{n \to +\infty} \big( \int_{C_2(\Delta')}^{C_2(\Delta_n)} f_{\Delta_n}(\hat{\Delta}  \mid l<\hat{\Delta}^{(1)}<u) \mathrm{d}\hat{\Delta} - \int_{C_1(\Delta')}^{C_1(\Delta_n)} f_{\Delta_n}(\hat{\Delta}  \mid l<\hat{\Delta}^{(1)}<u) \mathrm{d}\hat{\Delta} \big) \\
		= & 0.
	\end{align*}
	The first equation is proved. To prove the second equation, a completely analogous method can be used.
\end{proof}

\subsection*{B: Application of the C-UMAU confidence interval procedure in considered designs}
\subsubsection*{B.1. Application in the design with interim decision $D_1$ as described in Section 7.1} \label{A6}
For the design descried in Section 7.1, if the trial continues with the full population, i.e., $D_1(\mathcal{X}^{(1)}) = \{1,2\}$ we have
\begin{equation*}
	\{\mathcal{X}^{(1)}: D_1(\mathcal{X}^{(1)}) = \{1,2\}\} =\{\mathcal{X}^{(1)}: Z_{\mathcal{K}}^{(1)}> Z_* \} = \{\mathcal{X}^{(1)}: \frac{2\sigma}{\sqrt{n_1}} Z_* < \hat{\Delta}_{\mathcal{K}}^{(1)} < +\infty \}.
\end{equation*}
If the trial enriches to subpopulation 1, we have
\begin{align*}
	& \{\mathcal{X}^{(1)}: D_1(\mathcal{X}^{(1)}) = \{1\}\}  \\ = & \{\mathcal{X}^{(1)}: Z_{\mathcal{K}}^{(1)} \leq Z_*, Z_{1}^{(1)} \geq Z_{2}^{(1)} \} \\  = & \{\mathcal{X}^{(1)}: \sqrt{\frac{p_2}{p_1}} \hat{\Delta}_{2}^{(1)} \leq \hat{\Delta}_{1}^{(1)} \leq \frac{2\sigma}{p_1 \sqrt{n_1}}Z_* - \frac{p_2}{p_1}\hat{\Delta}_{2}^{(1)} \}.
\end{align*}
Similarly, if the trial enriches to subpopulation 2, we have
\begin{align*}
	& \{\mathcal{X}^{(1)}: D_1(\mathcal{X}^{(1)}) = \{2\}\} \\ = & \{\mathcal{X}^{(1)}: Z_{\mathcal{K}}^{(1)} \leq Z_*, Z_{1}^{(1)} < Z_{2}^{(1)} \}  \\ = & \{\mathcal{X}^{(1)}: \sqrt{\frac{p_1}{p_2}} \hat{\Delta}_{1}^{(1)} < \hat{\Delta}_{2}^{(1)} \leq \frac{2\sigma}{p_2 \sqrt{n_1}}Z_* - \frac{p_1}{p_2}\hat{\Delta}_{1}^{(1)} \}.
\end{align*}
According to Section 3.3, $D_1 \in \mathcal{D}$.

Then, we illustrate the calculation of the C-UMAU confidence intervals. If the trial continues with the full population, by Theorem 1(i), the conditional distribution of $\hat{\Delta}_{\mathcal{K}}$ given $D_1(\mathcal{X}^{(1)}) = \{1,2\}$ is
\begin{align*}
	& f(\hat{\Delta}_{\mathcal{K}} \mid D_1(\mathcal{X}^{(1)}) = \{1,2\}) \\
	= & f(\hat{\Delta}_{\mathcal{K}} \mid ({2\sigma}/{\sqrt{n_1}}) Z_*<\hat{\Delta}^{(1)}_{\mathcal{K}} < +\infty) \\
	= & \frac{1}{2\sigma/\sqrt{n}} \phi(\frac{\hat{\Delta}_{\mathcal{K}} - \Delta_{\mathcal{K}}}{2\sigma/\sqrt{n}}) \frac{1 - \Phi(\frac{({2\sigma}/{\sqrt{n_1}}) Z_* - \hat{\Delta}_{\mathcal{K}} }{\sqrt{n_2/n_1} \ 2\sigma/\sqrt{n}})}{1 - \Phi(\frac{({2\sigma}/{\sqrt{n_1}}) Z_* - \Delta_{\mathcal{K}}}{2\sigma/\sqrt{n_1}})}.
\end{align*}
Then, for any $\Delta_{\mathcal{K}, 0} \in \mathbb{R}$, we can construct the conditional level $\alpha$ uniformly most powerful unbiased test for the null hypothesis $H: \Delta_{\mathcal{K}} = \Delta_{\mathcal{K}, 0}$ using Step 1 in Section 4.2. Then, we can calculate the $100(1 - \alpha)\%$ C-UMAU confidence interval for $\Delta_{\mathcal{K}}$ using Step 2 in Section 4.2. If the trial enriches to subpopulation 1, by Theorem 1(i), the conditional distribution of $\hat{\Delta}_{1}$ given $D_1(\mathcal{X}^{(1)}) = \{1\}$ and the realization of $\hat{\Delta}_{2}^{(1)}$ is
\begin{align*}
	& f(\hat{\Delta}_{1} \mid D_1(\mathcal{X}^{(1)}) = \{1\}, \hat{\Delta}_{2}^{(1)}) \\
	= & f(\hat{\Delta}_{1} \mid \sqrt{\frac{p_2}{p_1}} \hat{\Delta}_{2}^{(1)} < \hat{\Delta}_{1}^{(1)} \leq \frac{2\sigma}{p_1 \sqrt{n_1}}Z_* - \frac{p_2}{p_1}\hat{\Delta}_{2}^{(1)} ) \\
	= & \frac{1}{2\sigma/\sqrt{p_1n_1 + n_2}} \phi(\frac{\hat{\Delta}_{1} - \Delta_{1}}{2\sigma/\sqrt{p_1n_1 + n_2}}) \frac{\Phi(\frac{\frac{2\sigma}{p_1 \sqrt{n_1}}Z_* - \frac{p_2}{p_1}\hat{\Delta}_{2}^{(1)} - \hat{\Delta}_{1} }{\sqrt{n_2/(p_1 n_1)} \ 2\sigma/\sqrt{p_1n_1 + n_2}}) - \Phi(\frac{\sqrt{\frac{p_2}{p_1}} \hat{\Delta}_{2}^{(1)} - \hat{\Delta}_{1} }{\sqrt{n_2/(p_1 n_1)} \ 2\sigma/\sqrt{p_1n_1 + n_2}})}{\Phi(\frac{\frac{2\sigma}{p_1 \sqrt{n_1}}Z_* - \frac{p_2}{p_1}\hat{\Delta}_{2}^{(1)} - \Delta_{1}}{2\sigma/\sqrt{p_1n_1}}) - \Phi(\frac{\sqrt{\frac{p_2}{p_1}} \hat{\Delta}_{2}^{(1)} - \Delta_{1}}{2\sigma/\sqrt{p_1n_1}})}.
\end{align*}
Then, using Step 1 in Section 4.2, for any $\Delta_{1, 0} \in \mathbb{R}$, we can construct the conditional level $\alpha$ uniformly most powerful unbiased test for the null hypothesis $H: \Delta_{1} = \Delta_{1, 0}$. Subsequently, we can calculate the $100(1 - \alpha)\%$ C-UMAU confidence interval for $\Delta_{1}$  using Step 2 in Section 4.2 given $D_2(\mathcal{X}^{(1)}) = \{1\}$ and the observed $ \hat{\Delta}_{2}^{(1)}$. Since this procedure provides the C-UMAU confidence intervals for every realization of $ \hat{\Delta}_{2}^{(1)}$, it also constitutes the C-UMAU confidence interval for $\Delta_1$ given $D_1(\mathcal{X}^{(1)}) = \{1\}$. If the trial enriches to subpopulation $\mathcal{S}_2$, a completely analogous procedure can be applied to obtain the confidence interval for $\Delta_2$.

\subsubsection*{B.2. Application in the design with interim decision $D_2$ as described in Section 6 and Section 7.2} \label{A7}
If the trial continues with the full population $\mathcal{F}$, i.e., $D_2(\mathcal{X}^{(1)}) = \{1,2\}$ , the stage 1 sample mean difference in $\mathcal{F}$ satisfies $\Delta_* < \hat{\Delta}^{(1)}_{\mathcal{K}} < +\infty$. Hence, we have
\begin{equation*}
	\{\mathcal{X}^{(1)}: D_2(\mathcal{X}^{(1)}) = \{1,2\}\} = \{\mathcal{X}^{(1)}: \Delta_* < \hat{\Delta}^{(1)}_{\mathcal{K}} < +\infty\}.
\end{equation*}
When the trial enriches to subpopulation $\mathcal{S}_1$, i.e., $D_2(\mathcal{X}^{(1)}) = \{1\}$, the stage 1 sample mean differences satisfy $-\infty < \hat{\Delta}^{(1)}_{\mathcal{K}} \leq \Delta_*$ and $\max(\hat{\Delta}_{1}^{(1)}, \hat{\Delta}_{2}^{(1)}) = \hat{\Delta}_{1}^{(1)} > \Delta_*$. Using $\hat{\Delta}^{(1)}_{\mathcal{K}} = p_1\hat{\Delta}_{1}^{(1)} + p_2 \hat{\Delta}_{2}^{(1)}$, these constraints can be equivalently written as $\Delta_*<\hat{\Delta}_{1}^{(1)} \leq (\Delta_* - p_2 \hat{\Delta}_{2}^{(1)})/p_1$. Hence, we have 
\begin{equation*}
	\{\mathcal{X}^{(1)}: D_2(\mathcal{X}^{(1)}) = \{1\}\} = \{\mathcal{X}^{(1)}: \Delta_*<\hat{\Delta}_{1}^{(1)} \leq (\Delta_* - p_2 \hat{\Delta}_{2}^{(1)})/p_1\}.
\end{equation*}
Similarly, if the trial enriches to $\mathcal{S}_2$, we can obtain that
\begin{equation*}
	\{\mathcal{X}^{(1)}: D_2(\mathcal{X}^{(1)}) = \{2\}\} =  \{\mathcal{X}^{(1)}: \Delta_*<\hat{\Delta}_{2}^{(1)} \leq (\Delta_* - p_1 \hat{\Delta}_{1}^{(1)})/p_2\}.
\end{equation*}
Hence, $D_2 \in \mathcal{D}$.

Then, we illustrate the calculation of the C-UMAU confidence intervals. If the trial continues with the full population, condition on $D(\mathcal{X}^{(1)}) = \{1,2\}$, we can obtain the probability density of $\hat{\Delta}_{\mathcal{K}}$ by Theorem 1(i), which is given by
\begin{align*}
	& f(\hat{\Delta}_{\mathcal{K}} \mid D(\mathcal{X}^{(1)}) = \{1,2\}) \\
	= & f(\hat{\Delta}_{\mathcal{K}} \mid \Delta_*<\hat{\Delta}^{(1)}_{\mathcal{K}} < +\infty) \\
	= & \frac{1}{2\sigma/\sqrt{n}} \phi(\frac{\hat{\Delta}_{\mathcal{K}} - \Delta_{\mathcal{K}}}{2\sigma/\sqrt{n}}) \frac{1 - \Phi(\frac{\Delta_* - \hat{\Delta}_{\mathcal{K}} }{\sqrt{n_2/n_1} \ 2\sigma/\sqrt{n}})}{1 - \Phi(\frac{\Delta_* - \Delta_{\mathcal{K}}}{2\sigma/\sqrt{n_1}})}.
\end{align*}
Accordingly, for any $\Delta_{\mathcal{K}, 0} \in \mathbb{R}$, we can construct the conditional level $\alpha$ uniformly most powerful unbiased test for the null hypothesis $H: \Delta_{\mathcal{K}} = \Delta_{\mathcal{K}, 0}$ using Step 1 in Section 4.2. Then, we can calculate the $100(1 - \alpha)\%$ C-UMAU confidence interval for $\Delta_{\mathcal{K}}$ using Step 2 in Section 4.2. When the trial enriches to subpopulation $\mathcal{S}_1$, conditioning on both $D_2(\mathcal{X}^{(1)}) = \{1\}$ and the realization of $ \hat{\Delta}_{2}^{(1)}$, the density of $\hat{\Delta}_{1}$ follows Theorem 1(i), given by
\begin{align*}
	& f(\hat{\Delta}_{1} \mid D(\mathcal{X}^{(1)}) = \{1\}, \hat{\Delta}_{2}^{(1)} ) \\
	= & f(\hat{\Delta}_{1} \mid \Delta_*<\hat{\Delta}_{1}^{(1)} \leq (\Delta_* - p_2 \hat{\Delta}_{2}^{(1)})/p_1) \\
	= & \frac{1}{2\sigma/\sqrt{p_1n_1 + n_2}} \phi(\frac{\hat{\Delta}_{1} - \Delta_{1}}{2\sigma/\sqrt{p_1n_1 + n_2}}) \frac{1 - \Phi(\frac{\Delta_* - \hat{\Delta}_{1} }{\sqrt{n_2/(p_1 n_1)} \ 2\sigma/\sqrt{p_1n_1 + n_2}})}{1 - \Phi(\frac{\Delta_* - \Delta_{1}}{2\sigma/\sqrt{p_1n_1}})}.
\end{align*}
Then, using Step 1 in Section 4.2, for any $\Delta_{1, 0} \in \mathbb{R}$, we can construct the conditional level $\alpha$ uniformly most powerful unbiased test for the null hypothesis $H: \Delta_{1} = \Delta_{1, 0}$. Subsequently, we can calculate the $100(1 - \alpha)\%$ C-UMAU confidence interval for $\Delta_{1}$  using Step 2 in Section 4.2 given $D_2(\mathcal{X}^{(1)}) = \{1\}$ and the observed $ \hat{\Delta}_{2}^{(1)}$. Since this procedure provides the $100(1 - \alpha)\%$ C-UMAU confidence intervals for every realization of $ \hat{\Delta}_{2}^{(1)}$, it also constitutes $100(1 - \alpha)\%$ C-UMAU confidence interval for $\Delta_1$ given $D_2(\mathcal{X}^{(1)}) = \{1\}$. If the trial enriches to subpopulation $\mathcal{S}_2$, a completely analogous procedure can be applied to obtain the confidence interval for $\Delta_2$.

If the trial continues with the full population and the co-primary analysis is desired, the stage 1 sample mean difference in $\mathcal{S}_1$ satisfies $(\Delta_* - p_2 \hat{\Delta}_{2}^{(1)})/p_1 < \hat{\Delta}^{(1)}_{1} \leq +\infty$, i.e., $\{\mathcal{X}^{(1)}: D_2(\mathcal{X}^{(1)}) = \{1,2\}\} = \{\mathcal{X}^{(1)}: (\Delta_* - p_2 \hat{\Delta}_{2}^{(1)})/p_1 < \hat{\Delta}^{(1)}_{1} \leq +\infty \}$. By conditioning on both $D_2(\mathcal{X}^{(1)}) = \{1,2\}$ and the realization of $ \hat{\Delta}_{2}^{(1)}$, an analogous method to that described above can be applied to obtain the confidence interval for $\Delta_1$ given $D(\mathcal{X}^{(1)}) = \{1,2\}$. Similarly, in this case, the stage 1 sample mean difference in $\mathcal{S}_2$ satisfies $(\Delta_* - p_1 \hat{\Delta}_{1}^{(1)})/p_2 < \hat{\Delta}^{(1)}_{2} \leq +\infty$, i.e., $\{\mathcal{X}^{(1)}: D_2(\mathcal{X}^{(1)}) = \{1,2\}\} = \{\mathcal{X}^{(1)}: (\Delta_* - p_2 \hat{\Delta}_{2}^{(1)})/p_1 < \hat{\Delta}^{(1)}_{1} \leq +\infty \}$. By conditioning on both $D_2(\mathcal{X}^{(1)}) = \{1,2\}$ and the realization of $ \hat{\Delta}_{1}^{(1)}$, we can obtain the confidence interval for $\Delta_2$ given $D(\mathcal{X}^{(1)}) = \{1,2\}$.

\subsection*{C. Additional adaptive enrichment designs for the proposed confidence intervals} \label{A8}

To illustrate the generality of the proposed confidence interval method, we consider two trial designs previously proposed in the literature \citep{kimani2015estimation, kimani2018point}. We retain the notation introduced in Section 3. The first example is the design of \citet{kimani2015estimation}. They considered the patient population can be partitioned into two subpopulations, i.e., $\mathcal{K} = \{1,2\}$. Further, prior knowledge suggests subpopulation 1 is more likely to benefit from the treatment than subpopulation 2. \cref{fig:decision3} shows the interim decision rule in \citet{kimani2015estimation}. If $ \hat{\Delta}^{(1)}_{1} -  \hat{\Delta}^{(1)}_{\mathcal{K}} > \Delta_*$, the enrollment criteria will be restricted to subpopulation 1 in stage 2. In this case, by replacing $\hat{\Delta}^{(1)}_{\mathcal{K}}$ with $p_1\hat{\Delta}^{(1)}_{1}  + p_2\hat{\Delta}^{(1)}_{\{2\}} $, we can obtain $\{\mathcal{X}^{(1)}: D(\mathcal{X}^{(1)}) = \{1\}\} = \{\mathcal{X}^{(1)}: \hat{\Delta}^{(1)}_{1}  -  \hat{\Delta}^{(1)}_{\mathcal{K}} > \Delta_*\} = \{\mathcal{X}^{(1)}: \hat{\Delta}^{(1)}_{\{2\}}  +  \Delta_*/p_2 <\hat{\Delta}^{(1)}_{1} < +\infty\}$. Then,  the C-UMAU confidence interval for $\Delta_1$ can be calculated using similar procedure as in Section B.1 and B.2. When the trial continues with the full population $\mathcal{F}$, i.e., $ \hat{\Delta}^{(1)}_{1} -  \hat{\Delta}^{(1)}_{\mathcal{K}} \leq \Delta_*$. Note that $\hat{\Delta}^{(1)}_{{1} } -  \hat{\Delta}^{(1)}_{\mathcal{K}}$ is independent of $\hat{\Delta}^{(1)}_{\mathcal{K}}$. Hence, given $D(\mathcal{X}^{(1)} ) =  \mathcal{K}$, the conditional distribution of $\hat{\Delta}^{(1)}_{\mathcal{K}}$ is not altered by the constraints on $\hat{\Delta}^{(1)}_{{1} } -  \hat{\Delta}^{(1)}_{\mathcal{K}}$, and $\hat{\Delta}^{(1)}_{{1} } -  \hat{\Delta}^{(1)}_{\mathcal{K}}$ is an auxiliary statistic for estimating $\Delta_{\mathcal{F}}$. Then, using the fact that $f(\hat{\Delta}_{\mathcal{K}} \mid D(\mathcal{X}^{(1)} ) =  \mathcal{K}) = f(\hat{\Delta}_{\mathcal{K}} \mid \Delta_* < \hat{\Delta}^{(1)}_{1}  -  \hat{\Delta}^{(1)}_{\mathcal{K}}<+\infty) =  f(\hat{\Delta}_{\mathcal{K}})$, Then,  the C-UMAU confidence interval for $\Delta_{\mathcal{K}}$ can be calculated. Actually, in this case, C-UMAU confidence interval and the naive confidence interval for $\Delta_{\mathcal{K}}$ are the same.

We can formally include the auxiliary statistic in class $\mathcal{D^*}$ such that
\begin{align*}
	\{\mathcal{X}^{(1)}: D^*(\mathcal{X}^{(1)}) = \mathcal{K}^{(2)} \} = \{\mathcal{X}^{(1)}: L_{\mathcal{K}^{(2)}} < \hat{\Delta}^{(1)}_{\mathcal{K}^{(2)}}< U_{\mathcal{K}^{(2)}},\ L_{\mathcal{K}^{(2)}}' < V_{\mathcal{K}^{(2)}} < U_{\mathcal{K}^{(2)}}'\},
\end{align*}
where $L_{\mathcal{K}^{(2)}}$, $U_{\mathcal{K}^{(2)}}$, $L'_{\mathcal{K}^{(2)}}$, and $U'_{\mathcal{K}^{(2)}}$ are random variables that may also take values in $\overline{\mathbb{R}} = \mathbb{R} \cup \{-\infty,+\infty\}$. In addition, $L_{\mathcal{K}^{(2)}}$, $U_{\mathcal{K}^{(2)}}$, and $V_{\mathcal{K}^{(2)}}$ are assumed to be independent of $\hat{\Delta}^{(1)}_{\mathcal{K}^{(2)}}$. Then, we have
\begin{align*}
	& f(\hat{\Delta}_{ \mathcal{K}^{(2)} } \mid D^*(\mathcal{X}^{(1)}) =  \mathcal{K}^{(2)}, L_{\mathcal{K}^{(2)}} = l_{\mathcal{K}^{(2)}}, U_{\mathcal{K}^{(2)}} = u_{\mathcal{K}^{(2)}}) \\
	& f(\hat{\Delta}_{\mathcal{K}^{(2)}} \mid l_{\mathcal{K}^{(2)}} < \hat{\Delta}^{(1)}_{\mathcal{K}^{(2)}}< u_{\mathcal{K}^{(2)}},\ L_{\mathcal{K}^{(2)}}' < V_{\mathcal{K}^{(2)}} < U_{\mathcal{K}^{(2)}}') \\
	= & f(\hat{\Delta}_{\mathcal{K}^{(2)}} \mid l_{\mathcal{K}^{(2)}} < \hat{\Delta}^{(1)}_{\mathcal{K}^{(2)}}< u_{\mathcal{K}^{(2)}}) \\
	= & \frac{\phi(\frac{\hat{\Delta}_{\mathcal{K}^{(2)}} - \Delta_{\mathcal{K}^{(2)}}}{2\sigma/\sqrt{p_{\mathcal{K}^{(2)}}n_1 + n_2}})}{2\sigma/\sqrt{p_{\mathcal{K}^{(2)}}n_1 + n_2}} \frac{\Phi(\frac{l_{\mathcal{K}^{(2)}} - \hat{\Delta}_{\mathcal{K}^{(2)}} }{\sqrt{n_2/(p_{\mathcal{K}^{(2)}}n_1)} \ 2\sigma/\sqrt{n_{\mathcal{K}^{(2)}}}}) - \Phi(\frac{u_{\mathcal{K}^{(2)}} - \hat{\Delta}_{\mathcal{K}^{(2)}} }{\sqrt{n_2/(p_{\mathcal{K}^{(2)}}n_1)} \ 2\sigma/\sqrt{n_{\mathcal{K}^{(2)}}}})}{\Phi(\frac{l_{\mathcal{K}^{(2)}} - \Delta_{\mathcal{K}^{(2)}}}{2\sigma/\sqrt{p_{\mathcal{K}^{(2)}}n_1}}) - \Phi(\frac{u_{\mathcal{K}^{(2)}} - \Delta_{\mathcal{K}^{(2)}}}{2\sigma/\sqrt{p_{\mathcal{K}^{(2)}}n_1}})},
\end{align*}
which has the same expression as in Section 4.3. Hence, our method can be applied with adaptive enrichment designs with interim decisions in the extended class $\mathcal{D^*}$.

\begin{figure}[]
	\begin{center}
		\begin{tikzpicture}[node distance=1.5cm and 0cm, auto]
			\node (start) [decision, text width=10em] {$\hat{\Delta}^{(1)}_{1} -  \hat{\Delta}^{(1)}_{\mathcal{K}} > \Delta_*$};
			\node (outcome1) [block, below left=2cm and 0cm of start] {Enrol from $\mathcal{S}_1$};
			\node (outcome2) [block, below right=2cm and 0cm of start] {Enrol from $\mathcal{F}$};
			\path [line] (start) -- node [midway, left] {Yes \hspace{0.5em}} (outcome1);
			\path [line] (start) -- node [midway, right] {\hspace{0.5em}  No} (outcome2);
		\end{tikzpicture}
	\end{center}
	\caption{Schematic diagram of the interim decision rule in \citet{kimani2015estimation}}
	\label{fig:decision3}
\end{figure}
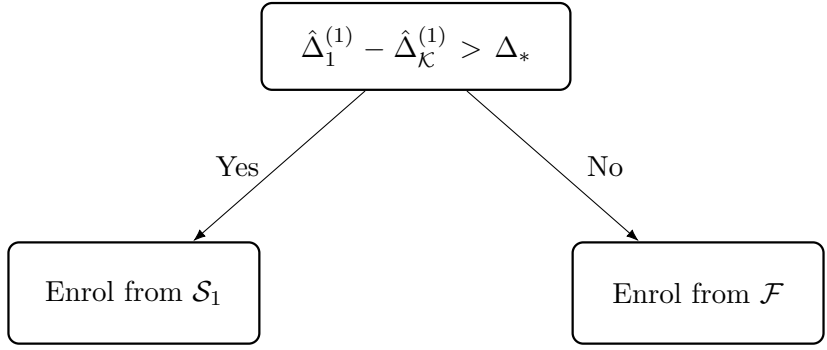

Another example is adaptive threshold  enrichment design proposed by \citet{kimani2018point}, where $k \ (k \geq 2)$ subpopulations are considered. They assume the true treatment effects in $\mathcal{S}_1$ to $\mathcal{S}_k$ is monotonically decreasing with $\Delta_{1} \geq \cdots \geq \Delta_{k}$. Define $[m] = \{1,\dots, m\} \ (m \leq k)$. In the interim analysis, the largest subgroup $\mathcal{S}_{[m]}$ whose stage 1 sample mean is greater than the threshold $\Delta_*$ is selected to continue in stage 2. Hence, we have 
\begin{equation*}
	\{\mathcal{X}^{(1)}: D(\mathcal{X}^{(1)}) = [m]\} =\{ \mathcal{X}^{(1)}: \hat{\Delta}_{[m]}^{(1)} > \Delta_*, \hat{\Delta}_{[m']}^{(1)} \leq \Delta_* \ (\text{for any }m' > m) \}
\end{equation*}
As shown in \citet{kimani2018point}, this selection event can also be expressed as 
\begin{equation*}
	\big\{ \mathcal{X}^{(1)}: \Delta_* <\Delta_{[m]}^{(1)} \leq \frac{p_{[m+1]}}{p_{[m]}}\Delta_* + \min\{ \frac{- p_{m+1} \hat{\Delta}_{m+1}^{(1)}}{p_{[m]}}, \frac{- \sum_{i = m+1}^{m+2} p_{i} \hat{\Delta}_{i}^{(1)}}{p_{[m]}}, \cdots, \frac{ - \sum_{i = m+1}^{k} p_{i} \hat{\Delta}_{i}^{(1)}}{p_{[m]}}   \} \big\}.
\end{equation*}
Accordingly, this design belongs to the class $\mathcal{D}$ where our method can be applied.


\clearpage

\renewcommand\thefigure{\arabic{figure}}    
\setcounter{figure}{0}


\end{document}